\documentclass[11pt]{scrartcl}

\usepackage{xy-format}
\usepackage{xy-theorem}
\usepackage{scrlayer-scrpage}
\usepackage{amsmath}
\usepackage{longtable}
\usepackage{tikz}
\usetikzlibrary{arrows.meta}
\clearpairofpagestyles
\newcommand{\C}{\mathbb{C}}
\newcommand{\Z}{\mathbb{Z}}
\newcommand{\R}{\mathbb{R}}
\newcommand{\Q}{\mathbb{Q}}
\newcommand{\tS}{\tilde{S}}
\newcommand{\sfT}{\mathsf{T}}
\newcommand{\sfS}{\mathsf{S}}
\newcommand{\Comm}{\mathrm{Comm}}

\theoremstyle{xyplain}
\newtheorem{corollary}[theorem]{Corollary}
\newtheorem*{maintheorem}{Main theorem}

\title{Genus-One Rigidity of Bounded Virasoro Pairings for 3D Gravity}
\author[1]{Xingyang Yu}
\affil[ ]{\href{mailto:xingyang.jason.yu@gmail.com}{\nolinkurl{xingyang.jason.yu@gmail.com}}}
\affil[1]{Department of Physics, Virginia Tech, Blacksburg, VA 24061, USA}
\date{}

\begin{document}
\maketitle

\begin{abstract}
We study whether varying the topological boundary of the doubled Virasoro TQFT
can produce distinct torus pairings for an ensemble interpretation of 3D gravity.
In the ordinary nondegenerate sector, every bounded pairing that commutes with
the modular transformations is scalar, and vacuum normalization selects the
diagonal pairing. The proof identifies the modular transformations with the
even Weil representation and applies the Cowling--Steger lattice-restriction
theorem. Under the ordinary block-diagonal vacuum--continuum ansatz, positivity
and the two vacuum marginals also make a Borel-measure pairing bounded, without
assuming absolute continuity; modular invariance then forces the same diagonal
pairing. These two classes therefore contain no distinct torus pairings to
average. A nontrivial Virasoro boundary ensemble must involve genus-one data
outside these classes, additional sectors or vacuum--continuum couplings, or
information invisible in the torus pairing. Elementary proofs in explicit
regular sectors explain the rigidity mechanisms, and numerical calculations
provide independent checks.
\end{abstract}

\newpage
\tableofcontents

\section{Introduction}\label{sec:intro}

A recurring puzzle in AdS$_3$/CFT$_2$ is that a gravitational path integral can
look more like an ensemble average than the partition function of a single
boundary theory~\cite{2006.08648,2203.06511,2308.03829,2407.02649}. The
Maloney--Witten sum over the known smooth genus-one saddles, for example, cannot
be interpreted as a torus trace over a Hilbert space~\cite{0712.0155}. Narain
theories provide a controlled counterpoint: their genus-one partition function,
averaged over moduli space, is reproduced by a bulk abelian TFT
sum~\cite{2006.04855,2006.04839,Yu:2026gdf}. These examples
motivate a basic question in ensemble holography: when a bulk path integral
computes an average, what precisely is being averaged on the boundary?

Symmetry TFT (SymTFT)\footnote{We follow the standard usage of each community:
TQFT for the Virasoro topological theory (QG/holography), and TFT in SymTFT
(QFT/condensed matter).} gives a concrete way to
formulate this question\footnote{See, e.g.,
\cite{Reshetikhin:1991tc,Turaev:1992hq,Barrett:1993ab,hep-th/9812012,
hep-th/0204148,Kirillov:2010nh,1008.0654,1012.0911,Kitaev:2011dxc,Fuchs:2012dt,
1212.1692,1412.5148,Kong:2014qka,Kong:2017hcw,Heckman:2017uxe,
Freed:2018cec,Thorngren:2019iar,Gaiotto:2020iye,Kong:2020cie,
Apruzzi:2021nmk,Freed:2022qnc,Kaidi:2022cpf,Antinucci:2022vyk,
Schafer-Nameki:2023jdn,2306.11783,Bhardwaj:2023kri,Baume:2023kkf,
Yu:2023nyn,2401.06128,2401.10165,DelZotto:2024tae,Argurio:2024oym,
Franco:2024mxa,Heckman:2024zdo,Gagliano:2024off,Cordova:2024iti,
Cvetic:2024dzu,Bhardwaj:2024igy,Bonetti:2024cjk,Apruzzi:2024htg,
2411.14997,Jia:2025jmn,Apruzzi:2025mdl,Heckman:2025lmw,Pace:2025hpb,
Luo:2025phx,Apruzzi:2025hvs,2510.06319,2603.12323} and references therein
for a partial list of foundational work, recent generalizations, and modern
overviews.}. In this framework,
a $D$-dimensional QFT with specified symmetry data can be represented by a
$(D{+}1)$-dimensional SymTFT on a slab~\cite{Gaiotto:2020iye,Freed:2022qnc,
Kaidi:2022cpf}. One side carries the physical boundary theory, while the other
carries a topological boundary condition; pairing the two produces an absolute
theory~\cite{1212.1692,Cvetic:2024dzu}. In this construction, changing the
topological boundary can change the global completion of the Hilbert space without
changing the Hamiltonian~\cite{Yu:2026gdf}. This suggests
a SymTFT interpretation of ensemble holography: keep the physical boundary and
the bulk symmetry data fixed, and average over topological boundary conditions,
or equivalently over the resulting absolute QFTs~\cite{2310.06012,2504.08724,
2510.03392,2511.04311}. In Abelian Chern--Simons theory,
topological boundary conditions are classified by Lagrangian subgroups of the
discriminant group~\cite{1008.0654,2203.09537}, the Abelian prototype of the
Lagrangian-algebra description. Related TQFT-gravity constructions organize
boundary ensembles through maximal gaugings or Lagrangian topological boundaries
and relate them to sums over bulk topologies~\cite{2310.13044,2405.20366}. In the
Narain example this picture is literal: the relevant
topological boundaries are parametrized by Narain moduli~\cite{Yu:2026gdf,
2606.15732}. Their torus partition functions are different, and averaging over
them reproduces the known ensemble result~\cite{Yu:2026gdf}.

The same mechanism is especially tempting for 3D gravity. The Virasoro TQFT has
been proposed as a topological organization of fixed-topology contributions to
AdS$_3$ gravity, and its doubled version carries the left- and right-moving
sectors relevant for a non-chiral 2D boundary theory~\cite{2304.13650,2401.13900}.
Recent triangulation-based and random-ensemble constructions sharpen this
fixed-topology relation~\cite{2507.11652,2507.12696,2407.02649,2506.19817,
2604.09396}; the geometric and open-sector boundary conditions appearing there
are distinct from the topological boundary conditions considered below.
We can then try to keep the physical boundary fixed, vary the topological
boundary condition of the doubled
Virasoro TQFT, and interpret the resulting theories as an ensemble of absolute
2D CFTs~\cite{Yu:2026gdf,2605.12590}. The question addressed in this paper is
whether this proposal already produces a nontrivial ensemble at genus one:
\begin{quote}
\emph{Can the doubled Virasoro TQFT admit non-diagonal topological boundary
conditions that lead to distinct torus partition functions?}
\end{quote}

At genus one the question has a precise formulation. Let $P$ and $Q$ label the
left- and right-moving nondegenerate Virasoro representations. Following the
proposed topological-boundary dictionary~\cite{Yu:2026gdf}, a topological
boundary determines multiplicity data $N(P,Q)$ and hence a torus partition
function
\begin{equation}
Z_N(\tau,\bar\tau)
 =\int_0^\infty dP\,dQ\,
 \bar\chi_P(\bar\tau)\,N(P,Q)\,\chi_Q(\tau).
\end{equation}
The diagonal boundary pairs equal representations,
$N(P,Q)=\delta(P-Q)$. Any different genus-one spectrum would require $N\neq I$,
either through a nonconstant diagonal weight or through off-diagonal support.
Since $Z_N$ must be modular
invariant, the integral operator with kernel $N$\footnote{Here
``kernel'' means an integral kernel, not the null space of an
operator: in the convention above, $N$ acts on a test function as
$(Nf)(P)=\int_0^\infty dQ\,N(P,Q)f(Q)$.}
must commute with the Virasoro $S$ and $T$ transformations. This is a necessary
genus-one condition. It is not, by itself, sufficient to construct a full
topological boundary, which requires further algebraic data and consistency
conditions, such as the Frobenius, Cardy and sewing conditions familiar from
rational CFT~\cite{hep-th/0204148}.

We work in the ordinary nondegenerate
sector\footnote{Here ``nondegenerate'' means the
Virasoro continuum with real momentum $P>0$ and no null states; ``ordinary''
excludes extra degenerate sectors and extended chiral algebras.}. Its modular
transformations are~\cite{Lacki:1990jb,2304.13650}
\begin{equation}
S(P,Q)=2\sqrt2\cos(4\pi PQ),
\qquad
T(P)=e^{2\pi i(P^2-1/24)}.
\end{equation}
\begin{samepage}
\noindent With this setup, the key genus-one rigidity statement is the following.
\begin{center}
\setlength{\fboxsep}{6pt}
\fbox{%
\begin{minipage}{0.91\textwidth}
\begin{maintheorem}
The bounded joint commutant of $S$ and $T$ is scalar:
\begin{equation}\label{eq:intro-commutant}
\Comm(S,T)=\C\,I,
\end{equation}
where $I$ is the identity and $\Comm(S,T)$ denotes the bounded joint commutant.
\end{maintheorem}
\end{minipage}}
\end{center}
\end{samepage}
\noindent Physically, this means that within the bounded class, modular invariance leaves
no nontrivial way to mix distinct left- and right-moving Virasoro
representations. The only bounded pairing is proportional to the diagonal one.
Concretely, we consider every
topological boundary whose genus-one pairing either defines a bounded operator
on the auxiliary label space, or is represented, under the ordinary
block-diagonal vacuum--continuum ansatz, by an entrywise-positive Borel measure
satisfying the two vacuum marginals. In the bounded case, the main theorem
reduces the pairing to a multiple of $\delta(P-Q)$, and vacuum normalization
fixes the coefficient to one. In the positive-measure case, the two marginals
first imply boundedness by a weighted Schur test, after which the same theorem
applies. Hence no topological boundary in these classes can produce a
nontrivial torus pairing. The comparison with Narain is analytic. Individual
Narain lattice pairings are positive atomic measures on the continuous
charge-label space and need not define bounded operators~\cite{Yu:2026gdf}.
For Virasoro, the two vacuum marginals close this positive-measure loophole
within the ordinary block-diagonal ansatz by forcing such a pairing into the
bounded class. Any Virasoro ensemble not excluded here must use information not
captured by the torus pairing, additional sectors, vacuum--continuum couplings,
or genus-one data outside the two classes.

\begin{figure}[p]
\centering
\newcommand{\roadmapref}[1]{{\footnotesize\color{black!60}#1}}
\begin{tikzpicture}[
  roadbox/.style={draw=black!45, rounded corners=2.5pt, line width=0.55pt,
    align=center, inner xsep=5pt, inner ysep=6.5pt, font=\small,
    minimum height=1.48cm},
  roadroute/.style={roadbox, text width=5.25cm},
  roadrowtop/.style={minimum height=2.15cm},
  roadrowmiddle/.style={minimum height=1.75cm},
  roadrowbottom/.style={minimum height=2.15cm},
  roadgate/.style={roadroute, fill=xysecrefblue!7},
  roadpositive/.style={roadroute, fill=xycitered!5},
  roadresult/.style={roadbox, fill=xyeqrefgreen!7,
    draw=xyeqrefgreen!65!black},
  roadsupport/.style={roadbox, text width=3.72cm, minimum height=1.72cm,
    fill=black!2, draw=black!28, font=\footnotesize},
  roadarrow/.style={-{Latex[length=2.2mm,width=1.5mm]},
    line width=0.7pt, draw=black!65}
]
  \node[roadroute, roadrowtop, fill=black!2] (candidate) at (-3.10,0)
    {\textbf{Candidate torus data} \quad $N(P,Q)$\\[-1pt]
     from an ordinary boundary\\[2pt]
     \roadmapref{physical input: Sec.~\ref{sec:vtqft}; formulation: Sec.~\ref{sec:torus2comm}}};
  \node[roadpositive, roadrowtop] (positive) at (3.10,0)
    {\textbf{Positive measure kernel}\\[-1pt]
     possibly singular or unbounded\\[2pt]
     \roadmapref{scope: Sec.~\ref{sec:scope}; result: Sec.~\ref{sec:vacuum}}};

  \node[roadgate, roadrowmiddle] (tgate) at (-3.10,-2.45)
    {\textbf{$T$: level matching}\\[-1pt]
     $u=P^2=\theta+k$: mixing survives\\[2pt]
     \roadmapref{complete solution: Sec.~\ref{sec:levelmatch}}};
  \node[roadpositive, roadrowmiddle] (marginals) at (3.10,-2.45)
    {\textbf{Two vacuum marginals}\\[-1pt]
     row and column normalization\\[2pt]
     \roadmapref{physical normalization: Sec.~\ref{sec:vacuum}}};

  \node[roadgate, roadrowbottom] (sgate) at (-3.10,-4.90)
    {\textbf{$S$: global mixing}\\[-1pt]
     cosine transform over all momenta\\[2pt]
     \roadmapref{second modular gate: Sec.~\ref{sec:torus2comm}}};
  \node[roadpositive, roadrowbottom] (schur) at (3.10,-4.90)
    {\textbf{Weighted Schur test}\\[-1pt]
     $\lVert N\rVert\leq 1$\\[2pt]
     \roadmapref{tool: Sec.~\ref{sec:prelim-comm}; result: Sec.~\ref{sec:vacuum}; proof: App.~\ref{app:schurproof}}};

  \node[roadresult, text width=5.25cm, minimum height=2.15cm] (scalar) at (0,-7.60)
    {\textbf{Bounded rigidity}\\[-1pt]
     $\Comm(S,T)=\C I$\\[2pt]
     \roadmapref{tool: Sec.~\ref{sec:prelim-weil}; theorem: Sec.~\ref{sec:layer2a}; proof: App.~\ref{app:weilproof}}};
  \node[roadresult, text width=5.25cm, minimum height=1.70cm] (vacuum) at (0,-10.05)
    {\textbf{Vacuum normalization}\\[-1pt]
     fixes the scalar to $1$\\[2pt]
     \roadmapref{Sec.~\ref{sec:vacuum}}};

  \node[roadresult, text width=12.5cm, minimum height=1.70cm] (nogo) at (0,-12.25)
    {\textbf{Unique pairing in the bounded and positive vacuum-marginal classes:}
     $N=I$ (diagonal)\\[1pt]
     $\Longrightarrow$ \quad
     \textbf{no distinct torus pairings to average in these classes}\\[2pt]
     \roadmapref{headline: Sec.~\ref{sec:intro}; implications and limits: Sec.~\ref{sec:open}}};

  \draw[rounded corners=3pt, dashed, line width=0.55pt, draw=black!38]
    (-6.45,-13.75) rectangle (6.45,-16.35);
  \node[fill=white, inner xsep=5pt, font=\small] at (0,-13.75)
    {\textbf{Independent mechanisms and checks}
     \quad \roadmapref{not inputs to the full bounded theorem}};
  \node[roadsupport] (branches) at (-4.20,-15.15)
    {\textbf{Explicit branch mechanisms}\\[2pt]
     \roadmapref{Secs.~\ref{sec:A}--\ref{sec:hardedge}}\\[-1pt]
     \roadmapref{proofs: App.~\ref{sec:diag}--\ref{app:hardedgeproof}}};
  \node[roadsupport] (tame) at (0,-15.15)
    {\textbf{Tame tails and scope}\\[2pt]
     \roadmapref{Sec.~\ref{sec:tame}; proof: App.~\ref{app:proofs}}\\[-1pt]
     \roadmapref{simple-current guardrail: Sec.~\ref{sec:B}}};
  \node[roadsupport] (checks) at (4.20,-15.15)
    {\textbf{Numerics and open gap}\\[2pt]
     \roadmapref{checks: Sec.~\ref{sec:numerics}}\\[-1pt]
     \roadmapref{quantitative layer: App.~\ref{app:gap}}};

  \draw[roadarrow] (candidate.east) -- (positive.west);
  \draw[roadarrow] (candidate.south) -- (tgate.north);
  \draw[roadarrow] (tgate.south) -- (sgate.north);
  \draw[roadarrow] (positive.south) -- (marginals.north);
  \draw[roadarrow] (marginals.south) -- (schur.north);
  \draw[roadarrow] (sgate.south) -- (scalar.north west);
  \draw[roadarrow] (schur.south) -- (scalar.north east);
  \draw[roadarrow] (scalar.south) -- (vacuum.north);
  \draw[roadarrow] (vacuum.south) -- (nogo.north);
\end{tikzpicture}
\caption{Reader's roadmap for the genus-one no-go and for the paper.  The gray
navigation text in each main box gives the subsection that develops the step and, where
appropriate, the appendix containing its complete proof.  Blue boxes mark the
bounded modular route, rose boxes the positive-measure upgrade, and green boxes
the shared rigidity and normalization conclusions; the neutral box gives their
common starting point.  The dashed lower band is a navigation guide to
independent concrete mechanisms, scope guardrails and numerical checks, not a
list of premises for the full bounded-commutant theorem.}
\label{fig:reader-roadmap}
\end{figure}

Because the argument behind this conclusion is technically involved,
Figure~\ref{fig:reader-roadmap} provides a roadmap before we turn to the
details.  Its central bounded step is representation-theoretic: the Virasoro
$S$ and $T$ transformations are identified with the even Weil representation
of the metaplectic group. The quadratic phase is
the metaplectic shear, while the cosine transform is the even Fourier transform.
The Cowling--Steger lattice-restriction theorem then implies that the restriction
to the modular lattice remains
irreducible~\cite{CowlingSteger:1991,Bekka:lattices}, and Schur's lemma gives the
scalar commutant. The
representation-theoretic theorem is not new; the point is that it applies
directly to the Virasoro modular problem and yields the genus-one no-go.

To make the rigidity mechanism concrete, we also give elementary proofs for several explicit classes of candidate
boundary conditions. Their common mechanism is easy to state. In the
squared-momentum variable $u=P^2$, the $T$ transformation only detects $u$
modulo an integer and therefore permits apparent branches $u\mapsto u+n$. The
$S$ transformation mixes
the full continuum, and its different oscillation frequencies prevent these
branches from combining into a non-diagonal invariant. Finite-dimensional
truncations provide an independent numerical check: in every tested basis the
joint commutant is 1D and aligned with the identity. These
computations corroborate the theorem but are not used as its proof.

The problem is the continuous counterpart of the classification of modular
invariants in rational CFT~\cite{Cappelli:1987xt,math/9902064}, but the logic is
different. Here the spectrum is continuous, and the result is a rigidity theorem
for the nondegenerate Virasoro modular representation rather than a
Cappelli--Itzykson--Zuber-type list. Earlier work derived the continuous modular
transformation of unitary $c\geq1$ Virasoro characters~\cite{Lacki:1990jb}, but
did not determine the bounded joint commutant relevant here. In the Virasoro-TQFT
setting, the diagonal continuum condensate has been constructed, while the
classification of possible additional condensates was left open~\cite{2412.11486}.
Related work using the same non-rational modular data addresses Verlinde-type
identities and open--closed duality rather than commutant
rigidity~\cite{2411.07285}. Our theorem resolves the bounded genus-one
modular-pairing part of this condensate-classification problem. To our knowledge,
this topological-boundary consequence of the
scalar bounded commutant has not previously been stated.

The conclusion concerns the torus pairing. Distinct topological boundaries
could still share the same diagonal torus partition function and differ at
higher genus or in their (higher-)categorical data\footnote{The possibility of
the same torus partition function but different categorical data is naturally
described using the language of extended field theory~\cite{1212.1692,
Lurie:2009keu}. Although most familiar for TQFTs, this language can also be used
schematically for CFTs: extendedness means retaining lower-codimension data and
is independent of absoluteness. Here ``absolute'' means that the partition
function on a closed 2D spacetime (codimension zero) is a number rather than a
vector supplied by the 3D bulk. In 2D, a spatial circle (codimension one) carries
a Hilbert space, whereas a point (codimension two) is assigned categorical data;
boundary conditions and line defects appear as objects, and their junctions as
morphisms~\cite{Freed:2022qnc}.}. Conversely, modular invariance of $N(P,Q)$ alone
does not establish that it comes from a topological boundary.

\paragraph{Organization of the paper.}
Section~\ref{sec:vtqft} reviews the Virasoro TQFT and its relation to 3D gravity.
Section~\ref{sec:setup} formulates the genus-one modular problem.
Section~\ref{sec:toolkit} introduces the operator-theoretic tools used to solve it.
Section~\ref{sec:results} presents bounded rigidity and its positive-measure
extension, followed by elementary mechanisms and numerical checks.
Section~\ref{sec:open} discusses the implications for ensemble holography and
the remaining directions. The detailed proofs are collected in the appendices.

\paragraph{Note added.}
The first version overstated the physical implications of the genus-one result,
as pointed out by Anatoly Dymarsky. The present version restricts the
interpretation to the bounded and ordinary positive vacuum-marginal classes
studied here. The mathematical results and proofs are unchanged.

\section{Virasoro TQFT and 3D gravity}\label{sec:vtqft}

In this section, we give a lightning review of Virasoro TQFT and its connection
to 3D gravity. The starting point is familiar from 2D CFT. The OPE and Virasoro Ward identities
expand a correlator in products of holomorphic and antiholomorphic conformal
blocks. Once the central charge, the external and intermediate representation
labels, the moduli of the punctured surface and a channel are specified, the blocks are fixed
kinematical functions. The spectrum, OPE coefficients and left--right pairing
determine how these functions are combined in a particular CFT. The Virasoro
TQFT packages the universal chiral part into a 3D theory; it does not by itself
choose a 2D CFT~\cite{2304.13650}. This distinction is useful below, because our
result constrains possible left--right pairings without attempting to
reconstruct all CFT correlators.

We work at
\begin{equation}\label{eq:cb}
c=1+6Q_L^2,\qquad Q_L=b+b^{-1},\qquad b>0,
\end{equation}
so $c\ge 25$. The ordinary nondegenerate Virasoro representations form a
continuum labeled by a real momentum $P>0$, with conformal weight
$h_P=Q_L^2/4+P^2$. There are also degenerate representations $L_{m,n}$ at the
Kac weights. In particular, the vacuum is a degenerate module at $h=0$ and does
not belong to the real-$P$ continuum~\cite{hep-th/0104158,2304.13650}. This last
point will matter later: the bounded commutant problem is first posed on the
ordinary continuum, while the vacuum returns through the normalization
conditions imposed on physical boundary data.

To see where the 3D structure enters, cut a Riemann surface into pairs of pants.
A basis of Virasoro conformal blocks is obtained by assigning representations to
the internal circles. A different cutting gives a different basis, related to
the first one by fusion and braiding kernels. For irrational Virasoro symmetry,
the change of channel is governed by the Ponsot--Teschner fusion
kernel~\cite{math/0007097,hep-th/0104158}. In the Racah--Wigner normalization,
the fusion kernel is the Virasoro $6j$ symbol~\cite{2411.07285}. The
Moore--Seiberg consistency relations make these changes of basis compatible
with sewing~\cite{Moore:1988qv}. Together with the inner product supplied by the
quantization of Teichm\"uller space, these data provide the non-rational
modular-functor input to the Virasoro TQFT construction~\cite{2304.13650}. In this
construction a Riemann surface is assigned a space of conformal blocks, while a
3D cobordism acts by the corresponding linear operator.

This construction is not a rational TQFT in disguise. Its state spaces are
infinite-dimensional, its ordinary line labels $L_P$ form a continuum, and
fusion is implemented by an integral transform rather than a finite sum. There
is no finite semisimple modular tensor category of the familiar RCFT
kind~\cite{2304.13650,2412.11486}. Nevertheless, the operations that matter here
still exist: cutting and gluing compose operators, mapping classes act on the
block spaces, and line networks are transformed by fusion and braiding kernels.
The resulting state space is a space of conformal blocks, not the physical
Hilbert space of a particular 2D CFT. Likewise, ``topological'' describes the 3D
organizing theory, not the 2D CFT whose blocks appear.

The torus gives the simplest example. The nondegenerate character is
\begin{equation}
\chi_P(\tau)=\frac{q^{P^2}}{\eta(\tau)},\qquad q=e^{2\pi i\tau}.
\end{equation}
Under $T$ it acquires a phase determined by
$h_P-c/24=P^2-1/24$, while $S$ mixes the continuum through a cosine transform;
the precise kernels are given in \eqref{eq:ST}~\cite{Lacki:1990jb}. The vacuum
character is instead the difference of two analytically continued Verma
characters at imaginary momenta, with the subtraction removing the vacuum null
descendant. Its $S$-transform defines the strictly positive
vacuum-to-continuum kernel $\rho_0(P)$ in \eqref{eq:rho0}, also known as the
Cardy/Plancherel density. Unlike the nondegenerate modular kernels, this vacuum
row depends on $b$~\cite{2304.13650}.

The proposed relation to 3D gravity gives this construction its holographic
interest. Brown--Henneaux boundary conditions produce left- and right-moving
Virasoro symmetries~\cite{Brown:1986nw}. A prescription for constructing the
gravity partition function on a fixed hyperbolic topology from the quantization
of Teichm\"uller space and the gluing and surgery rules of the Virasoro TQFT was
proposed in~\cite{2304.13650} and subsequently developed and tested on wormholes
and knot complements~\cite{2401.13900}. In bottom-up terms,
Virasoro symmetry controls the universal boundary kinematics, while the TQFT
organizes changes of channel and the cutting and gluing of the 3D manifold.

At the same time, the Virasoro TQFT should not be confused with a complete
definition of the 3D gravity path integral. The latter must additionally specify
the saddles and topologies of 3D manifolds to include, the quotient by large
diffeomorphisms, the weights of different contributions, and a nonperturbative
completion. The TQFT supplies the fixed-hyperbolic-topology building blocks and
their gluing rules; a sum over topologies and a prescription for off-shell or
non-hyperbolic geometries require further
input~\cite{2304.13650,2407.02649}.

Multi-boundary geometries motivate an ensemble interpretation: a fixed CFT
factorizes over disconnected boundaries, whereas connected bulk wormholes
encode connected moments of CFT data~\cite{2405.13111}. At genus one, the modular
sum over known torus saddles fails to define a physical Hilbert-space
trace~\cite{0712.0155}; in spectral language its density is continuous and not
positive definite~\cite{1407.6008}. This motivates reading the gravity result as
statistical information about a family of boundary theories. Within the
diagonal-condensation setup, factorization has been checked explicitly for
representative two-boundary torus and genus-two
wormholes~\cite{2412.11486}. This ensemble interpretation is not an input to our
proof; it motivates the more concrete boundary question that we study.

A single Virasoro TQFT organizes holomorphic conformal blocks. For a non-chiral
2D CFT with $(c,\bar c=c)$, the gravity construction uses two copies, one
left-moving and one right-moving~\cite{2304.13650}. We denote this doubled 3D
theory by
\begin{equation}
\mathcal D=\mathsf{Vir}_c\boxtimes\overline{\mathsf{Vir}}_c ,
\end{equation}
whose ordinary continuum lines are pairs $L_P\boxtimes\bar L_Q$. This is the 3D
bulk whose topological boundary conditions enter the question below. Such a
topological boundary is an auxiliary boundary condition of the doubled theory,
not the asymptotic AdS$_3$ boundary on which the physical CFT lives.
Figure~\ref{fig:symtft-slab} summarizes the Virasoro analogue of the Narain
averaging mechanism proposed in~\cite{Yu:2026gdf} and makes explicit which data
are held fixed and which are averaged over.

\begin{figure}[tbp]
\centering
\resizebox{0.99\linewidth}{!}{%
\begin{tikzpicture}[
  slabedge/.style={draw=black!70, line width=0.65pt},
  slabhidden/.style={draw=black!60, densely dashed, line width=0.6pt},
  slabarrow/.style={-{Latex[length=2.4mm,width=1.6mm]},
    draw=black!75, line width=0.8pt},
  slablabel/.style={font=\small, align=center},
  slabnote/.style={font=\footnotesize, text=black!60, align=center}
]
  \fill[black!1]
    (-5.25,-0.90) -- (-5.25,1.30) -- (-4.15,2.35) -- (-4.15,0.15) -- cycle;
  \draw[slabedge]
    (-5.25,-0.90) -- (-5.25,1.30) -- (-4.15,2.35) -- (-4.15,0.15) -- cycle;
  \node[slablabel] at (-4.70,-1.55)
    {\textbf{Absolute 2D theory}\\[-1pt]
     {\footnotesize genus-one pairing $N_\alpha(P,Q)$}};

  \draw[slabarrow] (-3.65,0.72) -- (0.62,0.72)
    node[midway, above=4pt, font=\small] {decompress};

  \fill[xysecrefblue!24]
    (1.00,1.30) -- (2.35,2.65) -- (2.35,0.25) -- (1.00,-1.10) -- cycle;
  \fill[xycitered!12]
    (4.25,1.30) -- (5.60,2.65) -- (5.60,0.25) -- (4.25,-1.10) -- cycle;
  \fill[black!1]
    (1.00,1.30) -- (2.35,2.65) -- (5.60,2.65) -- (4.25,1.30) -- cycle;
  \fill[white] (1.00,-1.10) rectangle (4.25,1.30);

  \draw[slabedge] (1.00,-1.10) rectangle (4.25,1.30);
  \draw[slabedge] (1.00,1.30) -- (2.35,2.65) -- (5.60,2.65) -- (4.25,1.30);
  \draw[slabedge] (1.00,-1.10) -- (2.35,0.25);
  \draw[slabedge] (4.25,-1.10) -- (5.60,0.25) -- (5.60,2.65);
  \draw[slabhidden] (2.35,2.65) -- (2.35,0.25) -- (5.60,0.25);

  \node[slablabel, anchor=south] at (1.45,2.94)
    {{\footnotesize averaged over}\\[-1pt]
     $\langle\mathcal B_\alpha|$\\[-1pt]
     {\footnotesize auxiliary topological boundary}};
  \node[slablabel, anchor=west] at (5.72,1.22)
    {$|\mathrm{Phys}\rangle$\\[-1pt]
     {\footnotesize physical boundary}\\[-1pt]
     {\footnotesize\color{black!60} held fixed}};
  \node[slablabel, anchor=north] at (3.30,-1.34)
    {\textbf{Doubled Virasoro TQFT}\\[-1pt]
     $\mathcal D=\mathsf{Vir}_c\boxtimes\overline{\mathsf{Vir}}_c$\\[-1pt]
     {\footnotesize\color{black!60} held fixed}};
\end{tikzpicture}
}
\caption{Decompression of an absolute 2D theory into the doubled Virasoro TQFT
on an interval.  In the proposed Virasoro analogue of the Narain averaging
mechanism of~\cite{Yu:2026gdf}, the ensemble average would hold the physical
boundary $|\mathrm{Phys}\rangle$ and the bulk theory $\mathcal D$ fixed while
averaging over the auxiliary topological boundaries
$\langle\mathcal B_\alpha|$.  Each choice of $\alpha$ would produce an absolute
theory and, at genus one, a pairing
$N_\alpha(P,Q)$ with torus partition function $Z_{N_\alpha}$.  Modular
invariance of this pairing is necessary for the full topological boundary
condition, but does not reconstruct it.}
\label{fig:symtft-slab}
\end{figure}

Section~\ref{sec:setup} now formulates these genus-one data as a left--right
pairing and derives the modular constraints. Section~\ref{sec:toolkit} then
introduces the operator-theoretic tools used to solve them.

\section{The genus-one modular problem}\label{sec:setup}

We now formulate the genus-one constraint on a topological boundary of the doubled
Virasoro TQFT. The problem is the continuous analog of a familiar rational-CFT
question. In a rational CFT, a matrix $N_{ij}$ records which left-moving primary
$i$ is paired with which right-moving primary $j$. Here the labels are continuous,
so the matrix becomes a kernel $N(P,Q)$ pairing left- and right-moving
nondegenerate Virasoro characters~\cite{Yu:2026gdf}. One may picture $P$ as the row
label and $Q$ as the column label of an infinite continuous matrix. The diagonal
kernel pairs $P$ only with itself. Any different torus spectrum would require
$N\neq I$, either through a nonconstant diagonal weight or through off-diagonal
support.

Modular invariance subjects this candidate matrix to two separate tests. The
$T$ transformation checks the relative phase of each proposed left--right pair;
it is the level-matching test. The $S$ transformation changes the character basis
by a cosine integral transform; it tests whether the whole pairing is unchanged
after the two cycles of the torus are exchanged. These are requirements on a
candidate $N$, not automatic consequences of writing one down. The division of
labor will be important: $T$ permits many integer-shifted non-diagonal pairs,
whereas $S$ is what eliminates their general mixing.

Section~\ref{sec:torus2comm} makes these two tests precise. We then solve the
$T$ test completely in Section~\ref{sec:levelmatch}: each fixed $T$-phase contains
an infinite fiber of integer-spaced momenta, and a $T$-invariant operator may act
by an arbitrary bounded matrix within each fiber. Section~\ref{sec:scope} fixes the analytic conventions
and explains how positivity and the vacuum enter. Section~\ref{sec:results} then
asks which of the level-matched matrices also pass the global $S$ test.

\subsection{From the torus spectrum to modular commutation}\label{sec:torus2comm}

In a rational CFT, a multiplicity matrix acts on vectors whose components are
indexed by the primary representations. Here that discrete index is replaced by the
continuous Virasoro momentum $P>0$, so the corresponding auxiliary label space is
\begin{equation}
 \mathcal H_P=L^2(\R_+,dP),\qquad
 \langle f,g\rangle_{\mathcal H_P}
 =\int_0^\infty dP\,\overline{f(P)}g(P).
\end{equation}
An element $f(P)$ is simply a square-integrable coefficient profile over the primary
labels. The measure $dP$ is the character normalization in which the modular
$S$-kernel is unitary; it is not the density of states of a particular CFT.
Likewise, $\mathcal H_P$ is only the auxiliary space on which the continuous
modular transformations act, not the physical Hilbert space of a 2D CFT. The continuum
multiplicity data can then be represented by an operator $N$ on $\mathcal H_P$.
When this operator has a generalized integral kernel $N(P,Q)$, it acts as
\begin{equation}
(Nf)(P)=\int_0^\infty dQ\,N(P,Q)f(Q)
\end{equation}
and formally gives the partition function
\begin{equation}
Z_N(\tau,\bar\tau)
=\int_0^\infty dP\,dQ\,
\bar\chi_P(\bar\tau)\,N(P,Q)\,\chi_Q(\tau).
\end{equation}
The diagonal pairing is $N(P,Q)=\delta(P-Q)$. More singular nonnegative data will
be treated later as a Borel measure\footnote{A Borel measure on
$(0,\infty)^2$ assigns a nonnegative weight, countably additively, to the sets
generated from open subsets of the label space. This formulation includes both
ordinary kernel densities $n(P,Q)\,dP\,dQ$ and singular measures supported on
lower-dimensional sets, such as the diagonal pairing. See~\cite{Bogachev:2007}.}
$N(dP,dQ)$. Such a measure initially pairs compactly supported test functions;
it need not define a bounded operator on $L^2$. Boundedness is assumed for the
main theorem and derived for positive, vacuum-normalized measures in
Section~\ref{sec:vacuum}.

The nondegenerate character and its modular transformations are~\cite{Lacki:1990jb,
2309.11540}
\begin{equation}
\chi_P(\tau)=\frac{q^{P^2}}{\eta(\tau)},\qquad q=e^{2\pi i\tau},
\end{equation}
and
\begin{equation}\label{eq:ST}
T(P)=e^{2\pi i(P^2-1/24)},\qquad
S(P,Q)=2\sqrt2\,\cos(4\pi PQ).
\end{equation}
Indeed, $h_P-c/24=P^2-1/24$. Equation~\eqref{eq:ST} is the exact modular kernel
for the unpunctured nondegenerate Virasoro characters. Although this kernel is
independent of the central-charge parameter $b$ of Eq.~\eqref{eq:cb}, the
$b$-dependence reappears below once the vacuum sector is included.

Writing the partition function schematically as
$Z_N=\langle\chi,N\chi\rangle$, we formulate invariance under the two
generators at the kernel level by the conditions
\begin{equation}
T^*NT=N,\qquad S^*NS=N,
\end{equation}
which we impose directly on the pairing operator. Since $S$ and $T$ are
unitary, these conditions are equivalent to
\begin{equation}\label{eq:commuting-condition}
[N,T]=0,\qquad [N,S]=0.
\end{equation}
Under the boundary dictionary above, these are the necessary genus-one
conditions that we test. Whenever the character integral converges, they are
also sufficient for modular invariance of $Z_N$\footnote{The converse, from
invariance of the function $Z_N$ back to the operator equations, is never used
in this paper. In a rational CFT that step is standard: the finitely many
characters are linearly independent, so equal partition functions have equal
multiplicity matrices. The continuum analog would require that no two distinct
kernels $N$ produce the same function $Z_N$, with the vacuum character treated
separately because it lies outside the real-$P$ continuum.}.

The two equations in \eqref{eq:commuting-condition} play different physical
roles. Since $T$ is diagonal in the character label, $[N,T]=0$ compares the
left and right phases entry by entry. By contrast, $S$ sends each character into
an integral over the full continuum, so $[N,S]=0$ compares all entries of $N$
after a global change of basis. Passing the first test therefore need not imply
passing the second.

The $T$ condition, $[N,T]=0$, is the level-matching condition\footnote{For a closed
string, invariance of the torus amplitude under $\tau\to\tau+1$ requires
$L_0-\bar L_0\in\Z$, the level-matching condition~\cite{Polchinski:1998rq}. A
Virasoro character already sums descendants in integer steps, so at the level
of characters the condition reduces to equality of the two $T$-phases,
$P^2-Q^2\in\Z$.}. For
a sufficiently regular kernel\footnote{``Sufficiently regular'' means that
$N(P,Q)$ is an ordinary function and that the relevant integrals converge. Equation~\eqref{eq:Tkernel}
then holds point by point, apart from a possible set of zero measure. If
the kernel contains delta functions or similar singular terms, the equation is
instead read after inserting it into an integral. For example,
$N(P,Q)=\delta(P-Q)$ satisfies the equation because the factor in parentheses
vanishes at $P=Q$.}, it becomes explicit: since $T$ acts by multiplication, the commutator
$[N,T]$ has kernel $\bigl(T(Q)-T(P)\bigr)N(P,Q)$, and dropping the overall
constant phase the condition reads
\begin{equation}\label{eq:Tkernel}
\bigl(e^{2\pi iP^2}-e^{2\pi iQ^2}\bigr)N(P,Q)=0,
\end{equation}
so its support can connect $P$ and $Q$ only when $P^2-Q^2\in\Z$. The $S$
condition, $[N,S]=0$, imposes an additional constraint. Unlike $T$, the modular $S$
transformation mixes each momentum label with the entire continuum. A
level-matched kernel must therefore also remain invariant under this mixing.

The simplest non-diagonal example separates the two tests. Consider the formal
graph kernel
\begin{equation}\label{eq:unitshift-example}
N_1(P,Q)=\delta\!\left(Q-\sqrt{P^2+1}\right).
\end{equation}
It pairs a label $P$ with a label whose squared momentum is larger by one. Hence
the two $T$-phases differ by $e^{2\pi i}=1$, and $N_1$ passes level matching.
But it fails the $S$ test for an elementary reason. If
$f\in C_c^\infty(0,1)$, then
$N_1f(P)=f(\sqrt{P^2+1})=0$, so $SN_1f=0$, whereas
\begin{equation}
(N_1Sf)(P)
=2\sqrt2\int_0^1 dQ\,
\cos\!\left(4\pi\sqrt{P^2+1}\,Q\right)f(Q)
\end{equation}
is nonzero for a suitable $f$. Thus $[N_1,S]\neq0$: the cosine transform
detects the displacement that $T$ cannot see. The kernel in
\eqref{eq:unitshift-example} is only a test candidate, not an asserted boundary
condition. We do not claim that it is a bounded or measure-preserving graph
operator, so it is not an example under the hypotheses of the graph-unitary
theorem below. Its purpose is only to show why both modular generators are needed. The
diagonal case, obtained by replacing the shift $1$ by $0$, is the identity
operator and passes both tests.

\begin{remark}[genus one only]\label{lem:Zinv}
A kernel passing the genus-one test need not extend to a topological boundary:
the full boundary requires further algebraic data and consistency conditions,
including a Frobenius-algebra structure and the Cardy and sewing
conditions~\cite{hep-th/0204148}. Note that the theorems below can only
exclude non-diagonal candidates at genus one, and never certify a topological
boundary.
\end{remark}

\subsection{Solving the level-matching condition}\label{sec:levelmatch}

\paragraph{The phase-fiber picture.}
Set $u:=P^2$, the representation-label contribution to
$h_P=(c-1)/24+P^2$ rather than a separate physical energy. Writing
$u=k+\theta$, with $k\in\mathbb N_0$ and $0\le\theta<1$, gives
\begin{equation}\label{eq:Ttheta}
T(u)=e^{2\pi i(u-1/24)}
=e^{2\pi i(\theta-1/24)}.
\end{equation}
Thus $T$ sees the fractional part $\theta=u\bmod1$ and forgets the integer
$k$. For example,
\begin{equation}
T(0.3)=T(1.3)=T(2.3)=\cdots,
\qquad T(0.4)\neq T(0.3).
\end{equation}
The fixed-$T$-phase fiber of $u\mapsto e^{2\pi iu}$ is
$\mathcal F_\theta:=\{\theta,1+\theta,2+\theta,\ldots\}$; its equally spaced
labels form a semi-infinite sequence. The $T$ condition alone may mix $0.3$ with
$2.3$, but not $0.3$ with $0.4$. For a kernel, Eq.~\eqref{eq:Tkernel} becomes
$u-v\in\Z$, which is exactly the same-fiber condition.
The answer to level matching is therefore
\begin{equation}\label{eq:Tfiber-summary}
\boxed{\begin{gathered}
{}[N,T]=0:\quad \text{different phase fibers cannot mix;}\\
\text{within one fiber any bounded mixing is allowed.}
\end{gathered}}
\end{equation}
This is much weaker than diagonality. It is simply the continuous version of a
diagonal matrix having degenerate eigenvalue blocks.

\paragraph{Making the change of variables unitary.}
We now put this elementary label-space picture into the exact Hilbert-space
normalization. This step adds no new physical condition; it only changes from
the coordinate $P$ to $u=P^2$ without changing norms. Since the measure becomes
$dP=du/(2\sqrt u)$, preserving the $L^2$ norm requires the half-Jacobian
factor $(dP/du)^{1/2}=1/(\sqrt2\,u^{1/4})$ on wavefunctions, just as for the
reduced radial wavefunction in quantum mechanics:
\begin{equation}
\int_0^\infty|f(P)|^2\,dP
=\int_0^\infty\biggl|\frac{f(\sqrt u)}{\sqrt2\,u^{1/4}}\biggr|^2\,du.
\end{equation}
The resulting unitary change of variables and its inverse are
\begin{equation}
U:\mathcal H_P\longrightarrow\mathcal H_u:=L^2((0,\infty),du),\qquad
(Uf)(u)=\frac{f(\sqrt u)}{\sqrt2\,u^{1/4}},\qquad
(U^{-1}g)(P)=\sqrt2\,\sqrt P\,g(P^2).
\end{equation}
Conjugating $T$ confirms Eq.~\eqref{eq:Ttheta} at the operator level:
\begin{equation}
(UTU^{-1}g)(u)=e^{2\pi i(u-1/24)}\,g(u).
\end{equation}
The half-Jacobian factors cancel between $U$ and $U^{-1}$, so conjugation only
replaces $P^2$ by $u$.

\paragraph{The complete bounded-operator solution.}
A wavefunction restricted to the fiber $\mathcal F_\theta$ is the sequence
$f_\theta=(f(\theta+k))_{k\ge0}\in\ell^2(\mathbb N_0)$, one for each $\theta$,
and $T$ acts on the whole sequence as the single scalar
$e^{2\pi i(\theta-1/24)}$. Consequently every bounded matrix $N(\theta)$ acting
on that fiber commutes with $T$, whereas matrix elements between different
fibers do not. This is the complete solution for arbitrary bounded
operators, including those without pointwise kernels. In direct-integral
notation, with the unitary map $(Wf)(\theta)=f_\theta$,
\begin{equation}\label{eq:Tfibres}
\mathcal H_u\cong\int_{[0,1)}^{\oplus}\ell^2(\mathbb N_0)\,d\theta,
\qquad
\Comm(T)\cong\int_{[0,1)}^{\oplus}\mathcal B(\ell^2(\mathbb N_0))\,d\theta,
\end{equation}
where $\mathcal B(\ell^2)$ denotes the bounded operators on one
fiber\footnote{The field $\theta\mapsto N(\theta)$ is measurable, and
boundedness means $\operatorname*{ess\,sup}_\theta\|N(\theta)\|<\infty$.}.
Equation~\eqref{eq:Tfibres} says exactly what the phase-fiber picture already
said: $\theta$ labels the fibers, $\ell^2(\mathbb N_0)$ records amplitudes on
their integer-spaced components, and $\mathcal B(\ell^2)$ is the arbitrary
bounded mixing allowed inside one fiber. The direct integral is the rigorous notation for that picture, not
an additional physical constraint.

\paragraph{A convenient shift subclass.}
At this point the level-matching problem is completely solved. For the explicit
regular classes studied later, it is convenient to expand the action inside a
fiber in integer shifts. This expansion is only a concrete subclass of the
full bounded solution~\eqref{eq:Tfibres}; it does not exhaust $\Comm(T)$. Let
\begin{equation}\label{eq:tau}
(\tau_nf)(u)=\mathbf 1_{\{u+n>0\}}\,f(u+n),\qquad n\in\Z,\qquad \tau_0=I.
\end{equation}
The indicator truncates the shift at the endpoint $u=0$. This endpoint is the
``hard edge'' used in Section~\ref{sec:hardedge}. Each $\tau_n$ commutes with $T$.
Writing $D_m$ for multiplication by $m(u)$, the explicit branch operators studied
in Sections~\ref{sec:hardedge} and~\ref{sec:E} take the form
\begin{equation}\label{eq:Nform}
\begin{aligned}
N&=D_m+\sum_{n\neq0}D_{a_n}\tau_n,\\
(D_{a_n}\tau_nf)(u)&=a_n(u)\,\mathbf 1_{\{u+n>0\}}\,f(u+n).
\end{aligned}
\end{equation}
When the branch set is finite, the sum in \eqref{eq:Nform} is a finite sum and
no convergence question arises; Section~\ref{sec:E} specifies the regularity and
convergence conditions under which an infinite series is admitted. As a kernel,
\begin{equation}
N(u,v)=m(u)\delta(u-v)+\sum_{n\neq0}a_n(u)\delta(v-u-n).
\end{equation}
Here $m$ is the coefficient on the diagonal branch and $a_n$ is the coefficient on
the shifted branch $v=u+n$.

Concretely, within one fixed-$T$-phase fiber the operator
$D_{a_n}\tau_n$ populates exactly the $n$-th diagonal of the fiber matrix
$N(\theta)$: for $u=\theta+k$ one has
$(D_{a_n}\tau_nf)(u)=a_n(\theta+k)\,f_\theta(k+n)$, so the only populated
entries are $(N(\theta))_{k,k+n}=a_n(\theta+k)$, on the active domain
$\theta+k+n>0$. Expanding $N$ in the shifts $\tau_n$ is therefore the same
as expanding a matrix in its diagonals. For a general bounded matrix such an
expansion need not converge in operator norm, so \eqref{eq:Nform} is a concrete
regular subclass of $\Comm(T)$, not a parametrization of the full commutant.

To impose the second modular condition on these level-matched operators, let
us now conjugate $S$ by the same map $U$. Substituting $v=Q^2$ and
$dQ=dv/(2\sqrt v)$, we find
\begin{equation}
\begin{aligned}
(USU^{-1}g)(u)
&=\frac{1}{\sqrt2\,u^{1/4}}
\int_0^\infty dQ\;2\sqrt2\,\cos(4\pi\sqrt u\,Q)\,\sqrt2\,\sqrt Q\,g(Q^2)\\
&=\int_0^\infty dv\;\sqrt2\,\frac{\cos(4\pi\sqrt{uv})}{(uv)^{1/4}}\,g(v),
\end{aligned}
\end{equation}
so the transformed operator $\tS:=USU^{-1}$ has kernel
\begin{equation}\label{eq:Stilde}
\tS(u,v)=\sqrt2\,\frac{\cos(4\pi\sqrt{uv})}{(uv)^{1/4}}.
\end{equation}
The factor $(uv)^{-1/4}$ is simply the pair of half-Jacobian factors of $U$
acting on the two arguments of the kernel\footnote{A positive diagonal change
of normalization passes between the character and Plancherel conventions; see
Remark~\ref{rem:D}.}. With the normalization in Eq.~\eqref{eq:Stilde}, $\tS$ is
the standard half-line cosine transform, hence a real symmetric unitary
involution: $\tS=\tS^{*}=\tS^{-1}$ and $\tS^2=I$.

The contrast with $T$ is now explicit. On a fixed $\theta$-fiber, $T$ is only
the scalar $e^{2\pi i\theta}$ and cannot see how the fiber entries are mixed.
The operator $\tS$, however, has a nonzero oscillatory kernel connecting the full
half-line. The remaining problem is therefore: which matrices acting inside the
$T$-phase fibers remain unchanged under this global cosine transform?

Because the shift class does not exhaust $\Comm(T)$, we present two kinds of
rigidity results: the elementary theorems of Sections~\ref{sec:hardedge}
and~\ref{sec:E} treat the class~\eqref{eq:Nform} directly, while the
full-commutant theorem of Section~\ref{sec:layer2a} covers all bounded
operators at once by a different, global method.

\subsection{What exactly is being constrained?}\label{sec:scope}

We use three analytic realizations of the same continuous multiplicity matrix.
The headline theorem treats $N$ as an arbitrary bounded operator, without a
pointwise kernel.  The elementary branch theorems use the explicit
delta-supported form~\eqref{eq:Nform}, while the positivity result starts from
a possibly singular measure on $(P,Q)$.  Kernel and distribution conventions
are given in Section~\ref{sec:prelim-comm}; the local realization of the tame
infinite branch sum is proved in Appendix~\ref{app:proofs}.

For a kernel or measure, \emph{entrywise positive} means a nonnegative Borel
measure on $(0,\infty)^2$: every region of label space receives nonnegative
multiplicity.  It is not the operator condition
$\langle f,Nf\rangle\geq0$; only entrywise positivity is used here.

The vacuum requires separate care because it is not one of the real-$P$
continuum states.  Its null descendant makes the vacuum character a difference
of two analytically continued Verma characters.  Applying $S$ gives the
strictly positive vacuum-to-continuum row~\cite{2309.11540,2304.13650}
\begin{equation}\label{eq:rho0}
\rho_0(P)=4\sqrt2\,\sinh(2\pi bP)\sinh(2\pi P/b)>0\qquad(P>0).
\end{equation}
In the $u=P^2$ normalization this row is
\begin{equation}\label{eq:rrow}
\begin{aligned}
r(u)&=(U\rho_0)(u)=\frac{\rho_0(\sqrt u)}{\sqrt2\,u^{1/4}}\\
&=4u^{-1/4}\sinh(2\pi b\sqrt u)\sinh(2\pi\sqrt u/b)>0.
\end{aligned}
\end{equation}
It grows exponentially and therefore lies outside $\mathcal H_u$.  Thus an
arbitrary bounded operator on $\mathcal H_u$ does not automatically act on
$r$; the notation $Nr$ must be justified rather than assumed.

We use the vacuum in only two controlled ways.  First, once the bounded
commutant theorem gives $N=cI$, the equation $Nr=r$ is unambiguous and fixes
$c=1$.  Second, under the ordinary block-diagonal ansatz consisting of one
vacuum block of multiplicity one and the continuum block, with no extra
discrete or degenerate sectors, full
modular invariance gives the two measure marginals
\begin{equation}\label{eq:vacuum-marginals-scope}
Nr=r,\qquad rN=r.
\end{equation}
For an entrywise-positive measure, the weighted Schur test of
Proposition~\ref{prop:schur} makes $N$ a bounded contraction.  The marginals
alone do not imply diagonality; that still requires modular commutation.  The
origin and scope of the block ansatz are recorded in
Remark~\ref{rem:crossblock}.

\begin{definition}[bounded genus-one pairing]\label{def:shadow}
A \emph{bounded genus-one pairing} is a bounded operator $N$ on
$L^2(\R_+,dP)$, equivalently its unitary image on $L^2((0,\infty),du)$,
satisfying $[N,T]=[N,S]=0$ in the first realization, equivalently
$[N,T]=[N,\tS]=0$ in the second. It is \emph{vacuum-normalized} if in addition
$Nr=r$ and $rN=r$, whenever these equations are meaningful.
\end{definition}

The boundary interpretation remains one-way: any topological boundary whose
genus-one pairing lies in this bounded class must produce such a pairing, but a
modular-invariant pairing need not contain the Frobenius, Cardy or sewing data
of a boundary. Section~\ref{sec:results}
shows that every bounded genus-one pairing is scalar, $N=cI$;
vacuum normalization then selects the diagonal pairing $N=I$.

\section{Operator-theoretic preliminaries}\label{sec:toolkit}

The bounded genus-one obstruction studied below uses three tools from operator
theory. Throughout this section, the $L^2$ spaces are auxiliary spaces of
Virasoro representation labels, not physical CFT Hilbert spaces. The commutant
and direct-integral dictionary solves the $T$-phase condition in
Section~\ref{sec:levelmatch}. Distributional kernels and the Schur test control
the explicit branch operators of Sections~\ref{sec:hardedge}--\ref{sec:E} and
the positive measures of Section~\ref{sec:vacuum}, respectively. Finally, the
even Weil representation and lattice
restriction give the full bounded commutant in Section~\ref{sec:layer2a}.
Each tool is the continuum version of a fact familiar from quantum mechanics or
linear algebra. A reader fluent in operator theory may use this section as a
reference and proceed directly to Section~\ref{sec:results}; the only substantial
external input is the lattice-restriction theorem described below.

\subsection{Commutants, kernels and the Schur test}\label{sec:prelim-comm}

For a set $\mathcal S$ of bounded operators, the \emph{commutant}
$\Comm(\mathcal S)$ is the set of all bounded operators commuting with every
element of $\mathcal S$. The rigidity statements of this paper are statements
about commutants: the smaller the commutant of the modular data, the fewer
invariant pairings exist.

The model case is diagonal. For $D=\operatorname{diag}(d_1,d_2,\dots)$ acting
on $\ell^2$, the space of square-summable sequences,
\begin{equation}
[N,D]_{ij}=(d_j-d_i)\,N_{ij},
\end{equation}
so $[N,D]=0$ forces $N_{ij}=0$ whenever $d_i\neq d_j$, while entries within
one eigenvalue block are unconstrained. The commutant of a diagonal matrix
therefore consists of exactly the block-diagonal matrices over its eigenvalue
blocks.

A \emph{multiplication operator} is the continuum version of a diagonal
matrix: $(Mf)(x)=m(x)f(x)$ on $L^2$, for a fixed measurable function $m$. It
is bounded when $m$ is essentially bounded, and its spectrum is the essential
range of $m$~\cite{ReedSimon:1980}. The finite-dimensional eigenvalue blocks
are replaced by spectral fibers. Literal point level sets may have measure
zero, so the invariant statement is phrased through the spectral projections
of $M$: a commuting operator preserves the corresponding spectral
decomposition and does not connect distinct spectral values
~\cite{ReedSimon:1980,Dixmier:1981}.

When the spectral fibers are labeled by a continuous parameter $\theta$,
``block diagonal with an arbitrary matrix in each block'' is written as a
\emph{direct integral}~\cite{Dixmier:1981},
\begin{equation}
\mathcal H=\int^{\oplus}\mathcal H_\theta\,d\theta,\qquad
N=\int^{\oplus}N(\theta)\,d\theta,\qquad
\|N\|=\operatorname*{ess\,sup}_\theta\,\|N(\theta)\|.
\end{equation}
Concretely, $N$ is one operator $N(\theta)$ for each $\theta$; the field
$\theta\mapsto N(\theta)$ is measurable, and boundedness means one common norm
bound for almost every $\theta$. In our application all fibers are identified
with the same $\ell^2$ space, so measurability can be read entry by entry in its
standard basis. Nothing beyond this definition is used. The modular $T$
transformation is a multiplication operator, its spectral fibers are the
fixed-$T$-phase sectors of Section~\ref{sec:levelmatch}, and Eq.~\eqref{eq:Tfibres}
is the corresponding direct-integral solution of level matching.

We write operators as integral kernels,
$(Nf)(P)=\int dQ\,N(P,Q)f(Q)$, and allow singular kernels such as the
diagonal pairing $\delta(P-Q)$. Two clarifications make this safe. First, an
identity between such kernels is read in the sense of
distributions~\cite{Hormander:1990}: both sides are paired with smooth test
functions of compact support in both variables, and the identity means that
all pairings agree. This is the sense in which delta-supported kernels solve
the modular kernel equations. Second, a kernel need not define a bounded operator
on $L^2$, and boundedness is the analytically decisive property throughout
this paper; it must either be assumed or derived.

The workhorse criterion for deriving it is the \emph{Schur
test}~\cite{HalmosSunder:1978}. Let $K(u,v)\ge0$ be a nonnegative kernel and
let $w>0$ be a positive weight controlling its rows and columns:
\begin{equation}\label{eq:schurtest}
\int K(u,v)\,w(v)\,dv\le A\,w(u),\qquad
\int w(u)\,K(u,v)\,du\le B\,w(v).
\end{equation}
Then $K$ defines a bounded operator with $\|K\|\le\sqrt{AB}$. The proof is a
weighted Cauchy--Schwarz inequality: splitting
$K=\bigl(K\,\tfrac{w(v)}{w(u)}\bigr)^{1/2}\bigl(K\,\tfrac{w(u)}{w(v)}\bigr)^{1/2}$,
\begin{equation}
\begin{aligned}
|\langle g,Kf\rangle|
&\le\iint|g(u)|\,K(u,v)\,|f(v)|\,du\,dv\\
&\le\Bigl(\iint|g(u)|^{2}\,K(u,v)\,\tfrac{w(v)}{w(u)}\,du\,dv\Bigr)^{1/2}
\Bigl(\iint|f(v)|^{2}\,K(u,v)\,\tfrac{w(u)}{w(v)}\,du\,dv\Bigr)^{1/2}\\
&\le\sqrt{AB}\,\|g\|_2\,\|f\|_2.
\end{aligned}
\end{equation}
This display is the density version of the Schur test. The physical positive
kernel may instead be a singular Borel measure. Proposition~\ref{prop:schur}
and Appendix~\ref{app:sector} give the measure version by applying the same
Cauchy--Schwarz argument directly with respect to $N(du,dv)$. Under the
ordinary block-diagonal vacuum--continuum ansatz, the two vacuum marginals give
$A=B=1$ with the vacuum row $r$ as weight. Positivity and both marginals then
give boundedness; diagonality still requires modular commutation.

\subsection{The Weil representation and lattice restriction}
\label{sec:prelim-weil}

The metaplectic group $Mp(2,\R)$ is the double cover of $SL(2,\R)$. It
carries a distinguished unitary representation $\omega$ on $L^2(\R)$, the
\emph{Weil}, or oscillator,
representation~\cite{Weil:1964,LionVergne:1980,Folland:1989,HoweTan:1992}.
Its action is generated by three familiar operations. Denote the shear, the
dilation and the inversion by
\begin{equation}
n(t)=\begin{pmatrix}1&t\\0&1\end{pmatrix},\qquad
a(s)=\begin{pmatrix}e^{s}&0\\0&e^{-s}\end{pmatrix},\qquad
w=\begin{pmatrix}0&-1\\1&0\end{pmatrix}.
\end{equation}
For chosen lifts of these elements, one has, up to phases,
\begin{equation}\label{eq:weilaction}
(\omega(n(t))f)(x)=e^{i\pi tx^2}f(x),\qquad
(\omega(a(s))f)(x)=e^{s/2}f(e^{s}x),\qquad
\omega(w)f=\hat f,
\end{equation}
respectively, where $\hat f(x)=\int e^{-2\pi ixy}f(y)\,dy$ is the Fourier
transform. These are the unitary implementations of elementary linear
canonical transformations; the double cover records their metaplectic sign.
Under $x\to-x$ the representation splits into an even and an odd part, and
each part is irreducible~\cite{Folland:1989,HoweTan:1992}. The phases left
undetermined in Eq.~\eqref{eq:weilaction} are central: they multiply the
identity and therefore drop out of every commutator.

In this paper the modular pair lives in the even part:
Section~\ref{sec:layer2a} identifies $T$, up to a constant phase, with the
shear at $t=1$, and the transformed $S$ with the inversion, written there as
the sans-serif matrices $\sfT$ and $\sfS$.

A unitary representation is \emph{irreducible} if its only closed invariant
subspaces are the zero space and the whole space. Schur's lemma then takes
the same form as in finite group theory: a bounded operator commuting with
every operator of an irreducible unitary representation is a multiple of the
identity~\cite{Dixmier:1981,Folland:1989}. This is the step that converts
irreducibility into a commutant statement.

The nontrivial question is whether irreducibility survives restriction to a
discrete subgroup. A \emph{lattice} $\Gamma<G$ is a discrete subgroup of
finite covolume, such as $SL(2,\Z)<SL(2,\R)$~\cite{BHV:2008}. A matrix
coefficient of a representation $\pi$ is the function
$g\mapsto\langle\pi(g)v,v'\rangle$, and $\pi$ is \emph{square-integrable} (a
discrete-series representation) if its matrix coefficients are
square-integrable over $G$~\cite{Knapp:1986}. The theorem of Cowling and
Steger, in Bekka's formulation, states: for $G$ connected simple with
finite center, $\Gamma<G$ a lattice, and $\pi$ irreducible and not
square-integrable, the restriction $\pi|_\Gamma$ is
irreducible~\cite{CowlingSteger:1991,Bekka:lattices}. We use this theorem as
a black box, in the way index theorems are commonly cited; it is restated as
Theorem~\ref{thm:CoS}, and the required failure of square-integrability for
the even Weil representation is a one-line Gaussian overlap computed in
Appendix~\ref{app:weilproof}.

With these tools the proof of the main theorem can be stated in one chain. The
modular pair generates the lattice preimage of $SL(2,\Z)$ inside $Mp(2,\R)$ up
to central scalars. The even Weil representation is irreducible and not
square-integrable, so lattice restriction keeps it irreducible on that
preimage; Schur's lemma then makes the joint commutant of $\tS$ and $T$ scalar.
Section~\ref{sec:layer2a} applies this chain to the Virasoro operators, and
Appendix~\ref{app:weilproof} supplies the phase analysis and proof.

\section{Rigidity of the bounded genus-one pairing}\label{sec:results}

This section asks one question in increasingly concrete forms. After $T$ has
imposed level matching but left an arbitrary matrix inside each fixed-phase
fiber, can any non-diagonal mixing also survive $S$? We first answer this
globally for every bounded pairing. We then show that positivity and the two
vacuum marginals make a measure pairing bounded, extending the conclusion to
possibly singular positive kernels. The remaining results explain rigidity
through permutations, finite integer-shift branches and summable infinite tails.

The two logical chains used for the headline no-go are
\begin{equation}\label{eq:two-rigidity-routes}
\begin{aligned}
&N\ \text{bounded},\quad [N,S]=[N,T]=0
   &&\Longrightarrow\quad N=cI
   &&\overset{\text{vacuum}}{\Longrightarrow}\quad N=I,\\
&N\ \text{a positive measure},\quad \text{two vacuum marginals}
   &&\Longrightarrow\quad \|N\|\le1,
   &&\overset{[N,S]=[N,T]=0}{\Longrightarrow}\quad N=I.
\end{aligned}
\end{equation}
Thus the second line uses modular invariance only in its final implication. For the
no-go itself, only Theorem~\ref{thm:layer2a}, its vacuum-normalization remark and
Corollary~\ref{cor:vacpos} are needed. The graph, hard-edge and tame theorems are
independent explanations of how non-diagonal candidates fail in explicit
regular sectors. Complete proofs are in Appendices~\ref{app:weilproof},
\ref{app:sector} and~\ref{app:proofs}.

\subsection{The full bounded commutant}\label{sec:layer2a}

At the $T$-only level there is plenty of room: a bounded operator may mix the
labels $u=\theta+k$ arbitrarily within each fixed-phase fiber. The unit shift
$\tau_1$ is the simplest example. The question is whether several such mixings
could conspire to remain unchanged under the global $S$ transformation. The
answer is no, not merely for the shift class but for every bounded operator.

\begin{theorem}[full bounded commutant]\label{thm:layer2a}
In the symmetrized half-line realization, equivalently in the even Schrödinger
realization of the Weil representation,
\begin{equation}
\Comm(\tS,T)=\C\cdot I:
\end{equation}
the only bounded operator commuting with both $\tS$ and $T$ is scalar.
\end{theorem}

In physics language, there is no bounded deformation of the left--right
multiplicity matrix that preserves both modular transformations. The theorem
does not say that all boundary conditions are identical. It says that every
bounded genus-one pairing reduces to a scalar operator.

\paragraph{Identifying the modular representation.}
Let us set $x=\sqrt2\,P$ and identify $L^2(\R_+,dP)$ unitarily with the even
subspace of $L^2(\R,dx)$. Up to the constant Virasoro phase, $T$ becomes
multiplication by $e^{i\pi x^2}$, while $S$, equivalently $\tS$ in the
$u$-coordinate, becomes the even Fourier transform with half-line kernel
$2\cos(2\pi xy)$. Thus the Virasoro operators are the images, up to scalar
phases, of the two modular generators reviewed in
Section~\ref{sec:prelim-weil},
\begin{equation}
\sfT=\begin{pmatrix}1&1\\0&1\end{pmatrix},\qquad
\sfS=\begin{pmatrix}0&-1\\1&0\end{pmatrix}.
\end{equation}

Their chosen lifts generate the lattice preimage of $SL(2,\Z)$ up to central
scalars. The even Weil representation is irreducible and not
square-integrable, so the lattice-restriction theorem and Schur's lemma reviewed
in Section~\ref{sec:prelim-weil} give the scalar commutant stated above.
Appendix~\ref{app:weilproof} checks non-square-integrability and gives the
complete central-phase analysis.

\begin{remark}[vacuum normalization]
With the vacuum normalization of Section~\ref{sec:vacuum}, Theorem~\ref{thm:layer2a}
fixes the unique bounded genus-one pairing as the diagonal one: any bounded
$N$ with $[N,\tS]=[N,T]=0$ is scalar, and $Nr=r$ forces $N=I$.
\end{remark}

\begin{remark}[soft theorem and complementary mechanisms]
Theorem~\ref{thm:layer2a} is a direct application of the Cowling--Steger
lattice-restriction theorem~\cite{CowlingSteger:1991,Bekka:lattices} and
Schur's lemma, rather than a new result in representation theory. The new
content is the Virasoro identification and its physical no-go interpretation. Sections~\ref{sec:hardedge}--\ref{sec:E} independently prove rigidity
explicitly and elementarily for the finite-branch and tame regular sectors, with no
representation-theoretic input: they expose the hard-edge and high-energy mechanisms
directly and also cover finite-branch non-tame coefficients such as
$\cos(\beta\sqrt u)$.
\end{remark}

\subsection{Vacuum normalization and positivity}\label{sec:vacuum}

The bounded theorem also controls a larger class of positive multiplicity
kernels. Such a kernel may initially be singular or unbounded. If it satisfies
the two vacuum marginals, the weighted Schur test makes it a contraction, and
modular invariance then forces the diagonal pairing. The marginals can be
imposed directly or obtained conditionally from the ordinary block-diagonal
vacuum--continuum ansatz described below.

There are therefore two different logical orders. For an already bounded
pairing, modular rigidity is applied first: $N=cI$, and vacuum normalization then
fixes $c=1$. For a positive measure kernel, the two vacuum marginals may instead
be imposed first: they imply $\|N\|\le1$, after which modular rigidity gives
$N=I$. Positivity supplies boundedness, not diagonality; diagonality still comes
from the joint $S,T$ commutant.

The physical normalization includes the identity/vacuum sector. The vacuum is a
rigged degenerate sector: with $P_+=\tfrac i2(b+b^{-1})$ ($h=0$) and
$P_-=\tfrac i2(b-b^{-1})$ ($h=1$), the vacuum character is
$\chi_{\mathrm{vac}}=\chi_{P_+}-\chi_{P_-}$, and analytic continuation of $S$ gives the
identity row
\begin{align}
S_{\mathbf1,Q}
&=2\sqrt2\bigl[\cosh(2\pi(b+b^{-1})Q)-\cosh(2\pi(b-b^{-1})Q)\bigr]\notag\\
&=4\sqrt2\,\sinh(2\pi bQ)\sinh(2\pi Q/b)=\rho_0(Q),
\end{align}
i.e.\ the strictly positive row $r(u)$ defined in \eqref{eq:rrow}. Vacuum
normalization is the pair of conditions $Nr=r$ and $rN=r$ (the vacuum multiplicity is
$1$ and couples to itself); here $b$ re-enters through $P_\pm$. For a bare bounded
$L^2$ operator these equations are not automatically meaningful because
$r\notin L^2$. Once rigidity has given $N=cI$, however, the pointwise equation
$cr=r$ is unambiguous and fixes $c=1$.

For an entrywise-positive kernel we instead formulate $N$ as a positive Borel
measure. Under the ordinary block-diagonal vacuum--continuum ansatz described in
Remark~\ref{rem:crossblock}, $Nr=r$ and $rN=r$ become two equations of measures.
One may think of $r$ as a positive weight attached to every continuum label: one
equation fixes the total weighted flow out of each label, and the other fixes the
flow into it. Together they are precisely the row and column estimates needed by
a weighted Schur test. A positive distribution is a positive Radon measure, so
this formulation also covers singular positive kernels.

\begin{proposition}[positivity + vacuum normalization $\Rightarrow$ contraction]\label{prop:schur}
Let $N$ be a positive Borel measure on $(0,\infty)^2$ satisfying the two vacuum
equations as measures,
\begin{equation}\label{eq:vacmeasure}
\int_{v} r(v)\,N(du,dv)=r(u)\,du,\qquad
\int_{u} r(u)\,N(du,dv)=r(v)\,dv .
\end{equation}
Then the bilinear form $B(g,f)=\iint \overline{g(u)}\,f(v)\,N(du,dv)$, defined first
for $f,g\in C_c((0,\infty))$, satisfies $|B(g,f)|\le\|g\|_{2}\|f\|_{2}$; hence $N$
defines a unique bounded operator on $L^2((0,\infty),du)$ with $\|N\|\le1$.
\end{proposition}

\begin{corollary}[no unbounded positive escape]\label{cor:vacpos}
Let $N$ be as in Proposition~\ref{prop:schur} and suppose the bounded operator it
defines commutes with $T$ and $\tS$. Then $N=I$, and the measure is the diagonal
$\delta(u-v)\,du$. In particular there are no \emph{non-diagonal} unbounded or
singular positive kernels in this class: for entrywise-positive vacuum-normalized
modular-invariant kernels the class collapses to the diagonal measure.
\end{corollary}

The proposition supplies boundedness; the corollary uses modular invariance to
supply diagonality. Both equations in \eqref{eq:vacmeasure} are essential, and
neither positivity nor the Schur estimate alone makes a kernel diagonal.

The proofs of Proposition~\ref{prop:schur} and Corollary~\ref{cor:vacpos} are
collected in Appendix~\ref{app:sector}.

\begin{remark}[both equations are needed, and the commutant theorem is still needed]\label{rem:schursharp}
One-sided normalization is not enough: $N(du,dv)=r(u)\,du\otimes\mu(dv)$ with
$\mu\ge0$, $\int r\,d\mu=1$ satisfies the first equation of \eqref{eq:vacmeasure}
but gives the unbounded rank-one map $f\mapsto(\int f\,d\mu)\,r$ (as $r\notin L^2$).
Conversely, the Schur bound alone never forces diagonality: on a two-point fiber
with weights $r=(1,2)$ the doubly $r$-normalized matrix
$\bigl(\begin{smallmatrix}1/2&1/4\\[1pt]1/4&7/8\end{smallmatrix}\bigr)$
satisfies both equations, has norm exactly $1$, and is not diagonal. Positivity plus
normalization buys \emph{boundedness}; diagonality still comes from the scalar
commutant (Theorem~\ref{thm:layer2a}).
\end{remark}

\begin{remark}[where the measure-level equations come from]\label{rem:crossblock}
Equations \eqref{eq:vacmeasure} are the rigged vacuum--continuum \emph{cross-blocks}
of full genus-one modular invariance for ordinary boundary data. Writing the full
kernel as
$N_{\mathrm{full}}=1\oplus N$ on $\mathrm{vac}\oplus\mathrm{continuum}$ (vacuum
multiplicity $1$; no vacuum--continuum coupling, i.e.\ no non-identity $h=0$ or
$\bar h=0$ states) and the full $S$ with vacuum row/column $r$, the
vacuum--continuum blocks of $[N_{\mathrm{full}},S_{\mathrm{full}}]=0$ are exactly
\eqref{eq:vacmeasure}, the $T$ cross-blocks are trivially satisfied, and the
continuum block is $[N,\tS]=0$. This is a conditional rigged-block statement under
the ordinary block-diagonal ansatz (motivation for the hypotheses of
Proposition~\ref{prop:schur}, in the spirit of Remark~\ref{lem:Zinv}), not a
standalone theorem; boundary data with vacuum--continuum couplings or extra
sectors (the extra-sector and extension directions of Section~\ref{sec:open}) lie
outside it. For the positive class this also reverses the logical order used
elsewhere in this section:
the normalization may be imposed \emph{first} (its rows are monotone integrals of a
positive measure, so no convergence issue arises), and rigidity then follows.
\end{remark}

\begin{remark}[boundary-level scope beyond the proved kernel class]
Corollary~\ref{cor:vacpos} proves diagonal uniqueness for entrywise-positive
Borel-measure genus-one data under the stated block-diagonal vacuum ansatz,
provided both vacuum marginals and modular invariance hold. What remains
conjectural at the boundary level is whether every ordinary/tame topological
boundary at generic irrational $c>25$ is captured by this analytic class and
whether additional categorical extension mechanisms are absent. Here
\emph{ordinary/tame} means built from the continuous Virasoro line category and
its rigged identity sector, excluding arbitrary extra Virasoro spectra. Genus-one
data outside the bounded or positive-measure classes also remain open.
\end{remark}

\subsection{No permutation invariants}\label{sec:A}

In a rational CFT, a permutation invariant pairs each left-moving primary with one
right-moving primary, so its multiplicity matrix has exactly one nonzero entry in
each row and column~\cite{Cappelli:1987xt,math/9902064}. The continuum analog is
a measure-preserving relabelling $P\mapsto\varphi(P)$. It acts by composition,
$U_\varphi f=f\circ\varphi$, rather than by a general integral kernel. The
measure-preserving condition is essential: it is the continuum replacement for the
statement that a permutation matrix is unitary.

Could a non-diagonal pairing survive simply by relabelling the continuum? A
tempting $T$-compatible map is
\begin{equation}
P\longmapsto\sqrt{P^2+n},\qquad n\in\Z,
\end{equation}
because it changes $P^2$ by an integer and hence preserves the $T$-phase. But a
permutation invariant must also preserve the Virasoro Plancherel measure. The
weights along each fixed-$T$-phase fiber are strictly increasing, so no two
labels in that fiber can be exchanged. This restricted graph class is already
rigid under $T$; $S$ is not needed.

\begin{theorem}[graph-unitary rigidity]\label{thm:A}
On $\R_+$, use the Virasoro Plancherel/Cardy measure
\begin{equation}
d\mu(P)=\rho_0(P)\,dP.
\end{equation}
Let $\varphi$ be a measure-space automorphism of $(\R_+,d\mu)$ and
$U_\varphi f=f\circ\varphi$. If
$U_\varphi$ commutes with
$T(P)=e^{2\pi i(P^2-1/24)}$, then $\varphi(P)=P$ a.e.
\end{theorem}

Here is the mechanism. Commutation with $T$ first imposes only equality of phases,
\begin{equation}\label{eq:graphTphase}
 e^{2\pi i\varphi(P)^2}=e^{2\pi iP^2}
 \qquad\Longrightarrow\qquad
 \varphi(P)^2-P^2\in\Z
 \quad\text{a.e.}
\end{equation}
Thus $T$ alone cannot distinguish labels whose values of $P^2$ differ by an integer. In
the variable $u=P^2$, the map $G(u)=\varphi(\sqrt u)^2$ preserves
$\theta=u\bmod1$. Meanwhile the Cardy/Plancherel measure becomes
\begin{equation}\label{eq:graphweight}
 \rho_0(P)\,dP=w(u)\,du,
 \qquad
 w(u)=\frac{\rho_0(\sqrt u)}{2\sqrt u}.
\end{equation}
For a fixed $\theta$, the $T$-fiber consists of the points $u=k+\theta$ and carries
the atomic measure $\sum_k w(k+\theta)\delta_{k+\theta}$. A measure-preserving
$G$ must therefore permute these atoms without changing their weights.

The Virasoro density makes that impossible. Up to a positive constant,
\begin{equation}
 w(P^2)=\frac{\sinh(2\pi bP)\sinh(2\pi P/b)}{P},
\end{equation}
and
\begin{equation}\label{eq:graphmonotone}
 P\frac{d}{dP}\log w(P^2)
 =2\pi bP\coth(2\pi bP)
  +\frac{2\pi P}{b}\coth(2\pi P/b)-1>0,
\end{equation}
because $x\coth x>1$ for $x>0$. Hence $w(k+\theta)$ is strictly increasing with
$k$: no two atoms on the same $T$-fiber have the same mass. The only
weight-preserving permutation is therefore the identity, which proves
$G(u)=u$ and hence $\varphi(P)=P$ almost everywhere. Appendix~\ref{app:sector}
records the disintegration argument and the associated measure-zero details.

This also explains why only $T$ is needed in Theorem~\ref{thm:A}. A formal shift
$u\mapsto u+n$ passes the phase test in \eqref{eq:graphTphase}, but it changes the
Plancherel weight and is not an allowed permutation. A general $T$-invariant
operator is much less rigid: it may mix several values of $k$ inside each
$\theta$-fiber rather than merely relabel them. The modular $S$ transformation is
needed to rule out that mixing, as in Theorem~\ref{thm:layer2a}. The graph theorem
rules out every measure-preserving graph relabelling, not only the special
permutations suggested by simple-current constructions~\cite{Fuchs:2004dz},
discussed in Section~\ref{sec:B}.

\subsection{Hard-edge finite-branch rigidity}\label{sec:hardedge}

Could several $T$-allowed integer branches combine into an $S$-invariant kernel?
Start with the simplest deformation,
\begin{equation}
N=I+\varepsilon D_a\tau_1.
\end{equation}
It passes $T$, because $u$ and $u+1$ have the same phase. It fails at the
physical endpoint $u=0$. On the strip $0<w<1$, the term in the $S$-commutator
that would evaluate the branch at $w-1$ is absent, while its partner remains.
The two terms therefore cannot cancel. This is the hard-edge mechanism.

More generally, $T$-invariance confines a local kernel to integer branches
$v-u=n\in\Z$. On a fixed $T$-fiber $\{\theta+k\}_{k\ge0}$,
$D_{a_n}\tau_n$ occupies one matrix diagonal, so finitely many branches form a
finite-band matrix. The result below is an explicit regular-sector theorem,
independent of and not needed for the full bounded theorem. A finite-branch
kernel here means \eqref{eq:Nform} with a finite branch set and a scalar diagonal,
\begin{equation}\label{eq:Nfinite}
N=c_0 I+\sum_{n\in F}D_{a_n}\tau_n,\qquad F\subset\Z\setminus\{0\}\ \text{finite},\ c_0\in\C.
\end{equation}
The scalar-diagonal hypothesis is part of the statement: a variable diagonal can
cancel off-diagonal commutators before the branches are removed, and is treated only
in the tame theorem (Theorem~\ref{thm:E}).

For a single branch, the commutator can be read directly from the kernel:
\begin{equation}\label{eq:hardedge-single-main}
\begin{split}
[D_a\tau_n,\tS](u,w)
={}&a(u)\mathbf 1_{\{u+n>0\}}\tS(u+n,w)\\
&-a(w-n)\mathbf 1_{\{w-n>0\}}\tS(u,w-n).
\end{split}
\end{equation}
If $n>0$, restrict to $0<w<n$. The second term then vanishes because the shifted
argument lies beyond the endpoint, while the first term is nonzero wherever $a$
is nonzero, apart from the measure-zero zero set of the cosine kernel. Thus a
positive branch cannot commute with $\tS$. For $n<0$ the reflected strip
$0<u<|n|$ gives the same conclusion. This proves the single-branch case without
any differentiability or high-energy assumption.

\begin{theorem}[hard-edge finite-branch rigidity]\label{thm:C}
Let $N$ be as in \eqref{eq:Nfinite} and suppose $[N,\tS]=0$, read in the local
kernel/distributional sense of Section~\ref{sec:setup} (the coefficients
$a_n\in L^\infty_{\mathrm{loc}}$ need not define a bounded $L^2$ operator). Then
$a_n=0$ a.e.\ on its
active domain $\{u+n>0\}$ for every $n\in F$ (so $N=c_0 I$) under any one of:
\emph{(i)} a single branch ($|F|=1$) with $a_n\in L^\infty_{\mathrm{loc}}$;
\emph{(ii)} a one-sided set ($F\subset\Z_{>0}$ or $F\subset\Z_{<0}$) with each
$a_n\in L^\infty_{\mathrm{loc}}$;
\emph{(iii)} a mixed set with each $a_n\in L^\infty_{\mathrm{loc}}\cap C^{|F|-1}_{\mathrm{loc}}$.
\end{theorem}

Several branches do not repair the mismatch. Near the edge, the modular kernel
assigns branch $n$ the distinct oscillation frequency
$4\pi\sqrt{u+n}$. Expanding the commutator at the endpoint turns the possible
cancellation into a finite Vandermonde system with one node for each active
branch. Distinct branches give distinct nodes, so every coefficient must vanish.
Thus the physical endpoint converts the continuum commutator problem into
finite-dimensional linear algebra. These are frequencies of the modular kernel,
not new Hamiltonian frequencies. Appendix~\ref{app:sector} gives the Taylor,
Vandermonde, Fubini and local-regularity details.

\begin{remark}[exact edge vs.\ quantitative gap]
The hard-edge proof is an exact rigidity argument: it uses equality of the commutator
kernel on a positive-measure boundary strip near $u=0$ or $w=0$. It gives no
quantitative lower bound for high-energy localized commutators and does not address
the still-open uniform operator-norm gap question $\kappa_T>0$ (with $\kappa_T$
defined in Appendix~\ref{app:gap}); conversely the selected-window cancellations
discussed there are approximate high-energy probes, not exact vanishing of the kernel
on an open boundary strip.
\end{remark}

\begin{remark}[the rigidity edge is the physical edge]
Physically $u=P^2\ge0$ is the bottom of the nondegenerate continuum, and the vacuum
row $r(u)$ re-enters from below this same edge: the rigidity edge and the vacuum
edge coincide. The hard-edge proof moreover covers strictly more finite-branch cases
than the high-energy method (arbitrary $L^\infty$ single and one-sided branches,
including non-tame oscillatory coefficients); on the tame finite-branch overlap the same
off-diagonal vanishing is independently recovered by the high-energy
nonstationary-phase mechanism of Theorem~\ref{thm:E}, which does not probe the edge.
These are \emph{not} two independent proofs of the same class: the high-energy method
does not cover non-tame single branches such as $\cos(\beta\sqrt u)$.
\end{remark}

\subsection{Tame decomposable rigidity}\label{sec:tame}\label{sec:E}

The bounded theorem already excludes invariant mixtures of infinitely many
integer shifts when their sum defines a bounded operator. We now recover this
conclusion directly from the oscillation frequencies of the modular kernel for
slowly varying amplitudes with summable tails. This requires an argument beyond
the finite-branch proof, since finite-band matrices are not norm-dense in the
full operator algebra of a $T$-fiber.

At high energy, branch $n$ carries its own frequency
$4\pi\sqrt{u+n}$. A projection that follows one chosen frequency isolates that
branch, while summability permits the infinite sum to pass through the limiting
argument. The following definition records the regularity needed to make this
picture rigorous; it is a sufficient physical regularity class, not a proposed
description of every element of $\Comm(T)$.

\begin{definition}[tame class]\label{def:tame}
A coefficient $a:(0,\infty)\to\C$ is \emph{tame} if it is a slowly varying symbol of
order $0$ in $u$: for a fixed integer $J\ge1$,
\begin{equation*}\label{eq:T1}
\|a\|_{(J)}:=\max_{0\le j\le J}\ \sup_{u>0}\,(1+u)^{j}\,\bigl|\partial_u^j a(u)\bigr|\ <\ \infty.
\tag{T1}
\end{equation*}
Equivalently, in $\xi=\sqrt u$ the lifted coefficient $b(\xi)=a(\xi^2-n)$ obeys the
symbol bounds $|\partial_\xi^k b(\xi)|\le C_k\langle\xi\rangle^{-k}$ for $0\le k\le J$
(the ``no resonant high frequency in $\sqrt u$'' condition). The operator $N$ of
\eqref{eq:Nform} is \emph{tame} if $m$ and every $a_n$ are tame and the off-diagonal
tail is summable in sup norm,
\begin{equation*}\label{eq:T2}
\sum_{n\neq0}\|a_n\|_\infty<\infty.
\tag{T2}
\end{equation*}
\end{definition}

Condition \eqref{eq:T2} makes the series \eqref{eq:Nform} converge in operator norm.
$C_c^\infty$ and Schwartz coefficients are tame. The oscillatory coefficient
$a(u)=\cos(\beta\sqrt u)$ is not: it injects a resonant high frequency and lies
outside this projection argument. It is nevertheless ruled out in the
single-branch case by the hard-edge theorem. Theorem~\ref{thm:E} uses
\eqref{eq:T1} for $m$ and every $a_n$ together with \eqref{eq:T2}; no tameness is
needed in the finite-branch Theorem~\ref{thm:C}. A stronger Schwartz class is a
convenient sufficient example, but the proof uses only \eqref{eq:T1}--\eqref{eq:T2}.

\begin{theorem}[tame decomposable rigidity]\label{thm:E}
Let $N\in\Comm(T)$ be a decomposable operator of the form \eqref{eq:Nform} that is
tame in the sense of Definition~\ref{def:tame} (\eqref{eq:T1} for $m$ and every
$a_n$, and the $\ell^1$ tail \eqref{eq:T2}). If $[N,\tS]=0$, then $a_n\equiv0$ for
$n\neq0$ and $m$ is constant; i.e.\ $N=c\cdot1$. If, in addition, the whole operator
$N$ lies in any operator ideal containing no nonzero scalar identity (for instance the
compact operators), then $c=0$.
\end{theorem}

\paragraph{Mechanism.} The proof (Appendix~\ref{app:proofs}) is a high-energy
frequency projection. Multiplying the commutator kernel by $w^{1/4}$ and setting
$\xi=\sqrt w$, each branch $n$ oscillates at its own frequency
$\alpha_n(u)=4\pi\sqrt{u+n}$, while the diagonal and all ``reflected'' terms
oscillate at the common frequency $\beta(u)=4\pi\sqrt u$. A tracking-frequency
functional tuned to $\alpha_{n_0}$ extracts branch $n_0$ alone: every other term is
killed by nonstationary phase (Lemma~\ref{lem:E1}), and a bound uniform in the
branch label and the energy scale (Lemma~\ref{lem:E2}) lets the summable tail
\eqref{eq:T2} pass through the limit. Hence $a_{n_0}\equiv0$ for every
$n_0\neq0$; only then is $N$ diagonal, and the elementary diagonal lemma
(Lemma~\ref{lem:diag}, Appendix~\ref{app:sector}) forces $m$ constant.

This is a no-go for a physically natural class of candidate genus-one kernels
(ordinary/tame boundary data motivate such tail-controlled, positive,
measure-like classes, Section~\ref{sec:tame}) and an explicit regular-sector
mechanism. The full bounded theorem of
Section~\ref{sec:layer2a} is proved separately, in Appendix~\ref{app:weilproof}.

\begin{remark}[convention transfer]\label{rem:D}
The character ($S$, measure $dP$), Plancherel ($\rho_0\,dP$) and symmetrized ($\tS$)
realizations are related by smooth positive diagonal conjugations $D$; a branch shift
becomes $D\tau_nD^{-1}=(d(u)/d(u+n))\,\tau_n$, multiplying the branch coefficient by a smooth
slowly-varying (symbol-class) weight while leaving the phase frequencies
$\sqrt{u+n}$, the only quantities either proof uses, unchanged. Thus the theorems transport to the
corresponding image of the finite-branch or tame class in each convention. (The
statement is that commutation and rigidity are preserved under the relevant
conjugation, not that one fixed operator lies in the identical tame class or commutes
with all three realizations of $S$.)
\end{remark}

For the explicit regular classes, the same final vacuum step gives the following
two corollaries.

\begin{corollary}[scalar-diagonal, finite branch]\label{cor:A}
Let $N=c_0I+\sum_{n\in F}D_{a_n}\tau_n$ be a finite-branch kernel as in
Theorem~\ref{thm:C} (any tier (i)--(iii)) with $[N,\tS]=0$. Then $N=c_0I$, and the
vacuum normalization $Nr=r$ forces $c_0=1$; i.e.\ $N=I$.
\end{corollary}

\begin{corollary}[variable-diagonal, tame]\label{cor:B}
Let $N=D_m+\sum_{n\neq0}D_{a_n}\tau_n$ be tame (Definition~\ref{def:tame}) with
$[N,\tS]=0$. By Theorem~\ref{thm:E}, $N=D_m$ with $m$ constant; vacuum normalization
$Nr=r$ then gives $m\equiv1$, i.e.\ $N=I$. This is the more genuinely CIZ-like
statement.
\end{corollary}

\subsection{Structural input: absence of simple-current constructions}\label{sec:B}

Could the usual rational-CFT simple-current mechanism evade the kernel rigidity?
At generic $c>25$, not within the ordinary Virasoro line category. A
nondegenerate line fuses into a continuum rather than a single inverse object,
while a degenerate line generically produces several shifted summands. Only the
tensor unit is invertible. This observation complements the permutation and
hard-edge no-go results, but it is a conditional structural input rather than a
headline theorem. Nondegenerate $W$-type extensions and conformal embeddings
remain outside the class considered here.

\begin{remark}[no simple currents]\label{prop:B}
Assuming the standard generic Liouville/Virasoro (Ponsot--Teschner) fusion rules, at
generic irrational $c>25$ the only invertible simple object in the ordinary/tame
Virasoro line category is the tensor unit: a nondegenerate line $L_P$ is not
invertible (the continuous-series product decomposes as a direct integral
$L_P\otimes L_Q\sim\int_0^\infty dR\,N(P,Q;R)\,L_R$~\cite{math/0007097,hep-th/0104158},
not a single object), and a degenerate $L_{m,n}$ fuses with a generic line into $mn$
shifted summands~\cite{1406.4290}, so only $L_{1,1}$ is invertible (generic $b$). Hence
there is no simple current and no simple-current extension. This is conditional on
those fusion rules, not a theorem for a rigorously constructed Virasoro MTC, and it
does not address conformal embeddings, integer-spin chiral extensions, extra
degenerate sectors, or signed/complex non-kernel genus-one data.
\end{remark}

Throughout, ``local'' means \emph{bosonic} local: chiral currents carry integer
conformal weight $h\in\Z$\footnote{The same degenerate-weight argument run with
$M\in\tfrac12\Z$ (so that $E_{\deg}$ is correspondingly enlarged) applies to
$h\in\tfrac12\Z$, and hence to fermionic/super currents; we do not use it.
Fermionic/super extensions would in any case require a spin-statistics refinement
outside the present scope.}.

\begin{remark}[no degenerate integer-spin currents]\label{lem:Bdeg}
Write $x=b^2$; the positive-Kac weight
$h_{r,s}=\tfrac14[(1-r^2)x+(1-s^2)x^{-1}+2(1-rs)]$ ($r,s\in\Z_{>0}$) is an integer $M$
iff
\begin{equation}
(1-r^2)\,x^2+[2(1-rs)-4M]\,x+(1-s^2)=0 .
\end{equation}
For $(r,s)\neq(1,1)$ this is a nonzero integer-coefficient polynomial of degree $\le2$
(its outer coefficients $1-r^2,\,1-s^2$ vanish together only at $r=s=1$), with finitely
many positive roots, so the exceptional set
\begin{equation}
E_{\deg}:=\!\!\bigcup_{(r,s)\neq(1,1),\,M\in\Z}\!\!\bigl\{x>0:(1-r^2)x^2+[2(1-rs)-4M]x+(1-s^2)=0\bigr\}
\end{equation}
is countable (Lebesgue-null). \emph{If $b^2\notin E_{\deg}$}, then apart from the
identity no degenerate primary has integer weight, hence there is no degenerate-type
bosonic integer-spin current. A transparent (strictly stronger) sufficient condition is
that $\{1,b^2,b^{-2}\}$ be $\Q$-linearly independent (i.e.\ $b^2$ neither rational
nor a quadratic irrational), the sharp exclusion set being $E_{\deg}$ itself. Two
cautions: ``$c$ irrational'' means only $b^2+b^{-2}$ irrational and does \emph{not}
imply $b^2\notin E_{\deg}$; and $E_{\deg}$ is invariant under $b\leftrightarrow1/b$
(since $h_{r,s}(1/b)=h_{s,r}(b)$). The label $(1,-1)$, with $h_{1,-1}=1$, lies outside
the positive Kac range and enters only as the weight-one null subtracted in the rigged
identity sector (Section~\ref{sec:vacuum}).
\end{remark}

\begin{remark}[scope of the extension exclusion: controlled / excluded / separate]\label{rem:Walgebra}
The Virasoro algebra is \emph{not} maximal at generic $c$: the $W_N$ algebras (e.g.\
the spin-$3$ $W_3$ extension) are genuine local chiral extensions by integer-spin
($h=3$) \emph{nondegenerate} currents, untouched by Remark~\ref{lem:Bdeg}. We make no
maximality claim; the status of the various mechanisms is a trichotomy.
\begin{itemize}
\item \emph{Controlled} (Theorems~\ref{thm:C}, \ref{thm:E}): candidate genus-one
kernels motivated by normalizable/tame ordinary boundary data, absolutely
continuous with respect to the Plancherel line, with tail-controlled
off-diagonal weight, plus the rigged identity sector. These theorems do not supply
the Frobenius, Cardy or sewing data of an actual boundary.
\item \emph{Excluded by scope}: a discrete atomic $W$-type chiral extension (extended
vacuum with a discrete Virasoro summand and an integer-spin primary) and conformal
embeddings; these are \emph{outside} the object class of
Theorems~\ref{thm:C}/\ref{thm:E}, not ruled out by any Virasoro maximality statement.
\item \emph{Separate / open}: signed or complex generalized kernels, genus-one
data outside the positive measure-kernel class and the block-diagonal vacuum ansatz,
and the quantitative gap $\kappa_T$
(Appendix~\ref{app:gap}).
\end{itemize}
Heuristically, Ponsot--Teschner fusion is a continuous direct integral rather than a
finite sum~\cite{math/0007097,hep-th/0104158}, so a discrete $W$-generator is invisible
to a normalizable Plancherel kernel; this is motivation only: a rigorous categorical
version would need a measurable algebra-object formalism we do not claim.
\end{remark}

Together the remarks above remove the simple-current and degenerate-current mechanisms
in the generic bosonic ordinary/tame setting; they do \emph{not} remove nondegenerate
$W$-type chiral extensions or conformal embeddings, which lie outside the
continuous/tame line category and remain scope/open.

\subsection{Numerical corroboration}\label{sec:numerics}

The analytic results do not rely on numerics, but the computations provide useful
corroboration and guardrails. They played three roles: discovery of the
1D commutant and branch frequencies $\sqrt{u+n}$, failure of
the tested finite-shift high-energy Weyl-sequence candidates, and evidence for the still-open quantitative
gap $\kappa_T>0$ in regular symbol/tame regimes. The code and figures are
reproducible, with $54$ unit tests\footnote{A self-contained source archive,
\nolinkurl{virasoro-boundary-rigidity-arxiv-v1-code.tar.gz}, accompanies the arXiv
submission as ancillary material. It contains the package, scripts, tests,
dependencies and exact run instructions. The commutant-SVD table is produced by
\nolinkurl{scripts/symmetrized_commutant.py} (cross-checked against
\nolinkurl{scripts/tinvariant_commutant.py}); the high-energy figure by
\nolinkurl{scripts/highenergy_shift.py}, with the dense-scale and operator-norm
validations in \nolinkurl{scripts/dense_validation.py} and
\nolinkurl{scripts/symbol_fullline.py}; the full suite runs via \texttt{pytest -q}
($54$ tests).}.

\begin{figure}[!t]
\centering
\includegraphics[width=0.72\linewidth]{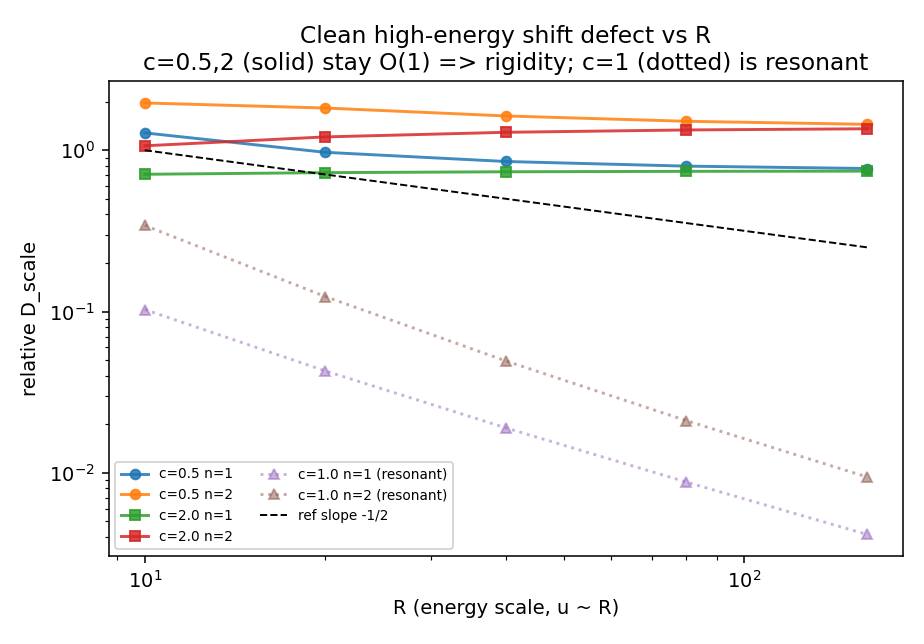}
\caption*{\textbf{Numerical illustration.} Scale-resolved defect of a single high-energy shift
vs.\ energy $R$. The non-resonant probes ($c=0.5,2$, solid) plateau at $O(1)$;
the resonant channel ($c=1$, dotted) decays.}
\end{figure}

\paragraph{Commutant SVD.} We discretize $u=P^2=k+\theta$ ($k=0,\dots,K$,
$\theta\in[0,1)$ on $A$ midpoint or Gauss--Legendre nodes). $T$-exact operators are
$\theta$-block-diagonal fiber matrices $M_{kl}(\theta)$, and $S$ is the symmetric
$b$-free kernel $\tS$. For the linear map $L:M\mapsto[N,\tS]$, the nullspace is
1D and equals the identity in every diagonal, banded and general
$T$-exact subspace tested. The table below shows the general subspace at $A=8$.
The character and symmetrized bases agree on the null dimension, as do midpoint and
Gauss--Legendre quadratures.

\begin{center}
\small
\begin{tabular}{@{}rrcc@{}}
\toprule
$K$ & $\dim M$ (general) & null count at $10^{-8},10^{-10},10^{-12}$ & $\sigma_2/\sigma_{\max}$ \\
\midrule
4  & 200  & $1/1/1$ & $0.851$ \\
6  & 392  & $1/1/1$ & $0.853$ \\
8  & 648  & $1/1/1$ & $0.845$ \\
10 & 968  & $1/1/1$ & $0.836$ \\
12 & 1352 & $1/1/1$ & $0.840$ \\
\bottomrule
\end{tabular}
\parbox{0.96\linewidth}{\small\textbf{Commutant-SVD check.} Symmetrized-basis $T$-exact commutant
of $S$ (general fiber subspace, $A=8$): the nullspace is 1D (the
identity; residual $0$, overlap $1.00000$) and the relative gap
$\sigma_2/\sigma_{\max}$ stays open over the tested range $K=4,\ldots,12$.}
\end{center}

\paragraph{High-energy probes.} In the localized, scale-resolved probes studied
here, a single high-energy shift $\tau_n$ is not an approximate symmetry: the
clean, non-resonant defects
$D_{\mathrm{scale}}(c=0.5,2)$ stay $O(1)$ as the energy $R\to\infty$; only the
resonant $c=1$ channel decays. Finite multi-branch combinations can cancel in
selected scale channels under an $L^2$/coefficient norm, but this does not survive a
dense-scale / wide-band validation. Under the operator (multiplier) norm the
reciprocal defect plateaus at a positive floor that grows with the band, consistent
with Lemma~\ref{lem:fullline}. We deliberately avoid the global Frobenius ratio,
since $\|\tS\|_F\sim\sqrt U$ diverges with the cutoff and manufactures a spurious
decaying trend.

\section{Discussion}\label{sec:open}

The nondegenerate Virasoro torus pairing is rigid in two nested classes. For an
arbitrary bounded operator, Theorem~\ref{thm:layer2a} gives
$\Comm(\tS,T)=\C I$, and vacuum normalization gives $N=I$. For a positive Borel
measure kernel, Proposition~\ref{prop:schur} derives boundedness from the two
vacuum marginals, so Corollary~\ref{cor:vacpos} again gives the diagonal measure.
The graph, finite-branch and tame theorems provide independent elementary
mechanisms in explicit regular sectors. Thus diagonality is not merely a convenient
ansatz: it is the only genus-one pairing in these classes compatible with the
nondegenerate modular data. The representation-theoretic step is a corollary of
Cowling--Steger; the new physics statement is the resulting obstruction to
manufacturing a non-diagonal torus ensemble from pairings in either class.

The Narain average clarifies the scope of this obstruction. Different lattice
data give distinct torus kernels, and averaging them reproduces the
Chern--Simons/$U(1)$ gravity answer~\cite{2006.04855,2006.04839,Yu:2026gdf}.
On a fixed continuous charge-label space, each lattice pairing is a positive
atomic measure and need not define a bounded $L^2$ operator. Scalarity of the
bounded commutant therefore does not by itself distinguish Virasoro from Narain.
The additional Virasoro conclusion comes from the vacuum. Under the ordinary
block-diagonal ansatz, the two vacuum marginals turn an entrywise-positive Borel
measure into a contraction, after which bounded rigidity forces the diagonal
pairing. A Narain-like Virasoro construction, if one exists, must therefore use
genus-one data outside both the bounded and positive vacuum-marginal classes or
use information not captured by the genus-one pairing.

Our result therefore answers the ordinary genus-one bounded-pairing part of the
\S5.2 proposal of~\cite{Yu:2026gdf}; it does not disprove the broader
Virasoro-TQFT/3D-gravity ensemble program. That program may depend on data not
visible in the torus partition function, on enlarged chiral sectors, or on genuinely
generalized boundary data.

Within the stated ordinary sector, every vacuum-normalized modular-invariant
pairing in the bounded or positive vacuum-marginal class is $I$, so any
normalized average of such pairings also returns $I$. The conditions apply to
the individual pairings being averaged. Signed, complex or non-kernel
functionals outside these classes remain possible.

This input/output distinction matters for gravity. A distributional
Maloney--Witten sum over modular images is an \emph{output} of the gravitational
average; its non-normalizability does not by itself provide distinct individual
Virasoro genus-one data to average. Such distinct inputs would have to evade the
stated bounded or positive-kernel hypotheses, or differ in information invisible at
genus one.

The theorem concerns a necessary torus-level condition. It neither constructs nor
classifies topological boundaries, and it does not impose the
Frobenius--Cardy--sewing or Virasoro $6j$ constraints~\cite{2411.07285}. The main
open directions are:
\begin{enumerate}
\item A rigorous operator-level proof or disproof of $\kappa_T>0$ (uniform
operator--symbol error control of Lemma~\ref{lem:fullline}, Appendix~\ref{app:gap}).
\item Signed or complex generalized kernels, non-kernel
Poincar\'e-series-type functionals, and arbitrary rigged kernels not controlled by
the hard-edge theorem. The positive measure-kernel case satisfying both vacuum
marginals and modular invariance is already closed.
\item Vacuum--continuum couplings, extra degenerate sectors, nondegenerate
$W$-type extensions and conformal embeddings outside the continuous/tame line
category.
\item The higher-genus and categorical upgrade: Frobenius algebra data,
Cardy/sewing consistency, Virasoro $6j$ symbols and the eventual classification
of actual topological boundaries.
\item The role of these additional data in AdS$_3$ ensemble holography.
\end{enumerate}

\addsec{Acknowledgements}
The author thanks Scott Collier, Elliott Gesteau, Yikun Jiang, and Justin Kulp for
helpful discussions. The author is particularly grateful to Anatoly Dymarsky for
pointing out an incorrect physical claim in the first version of this paper. The
author thanks Anthropic for providing access to Claude through its Max plan during
the early stages of this project. The numerical components of this work were
developed with assistance from GPT and Claude.
\href{https://openai.com/codex/}{OpenAI Codex} was used to assist with wording
and organization during revision of the manuscript. This work was partially supported
by NSF grant PHY-2310588.

\appendix
\section{Metaplectic proof of the bounded commutant}\label{app:weilproof}

We spell out the representation-theoretic input behind
Theorem~\ref{thm:layer2a}. Under $x=\sqrt2\,P$,
$L^2(\R_+,dP)\cong L^2_{\mathrm{even}}(\R)$ carries $T=\rho(\sfT)$ (multiplication by
$e^{i\pi x^2}$, up to the constant Virasoro phase) and
$\tS=\rho(\sfS)|_{\mathrm{even}}$ (the even Fourier transform), where
$\rho=\omega_{\mathrm{even}}$ is the even oscillator representation of $Mp(2,\R)$,
\begin{equation}
\sfT=\begin{pmatrix}1&1\\0&1\end{pmatrix},\qquad
\sfS=\begin{pmatrix}0&-1\\1&0\end{pmatrix}.
\end{equation}
(In the Lie-theoretic literature these matrices are written $n(1)$ and $w$.)
The oscillator representation is a genuine unitary representation of $Mp(2,\R)$ and
splits into the irreducible even and odd pieces~\cite{Folland:1989,Howe:1988,HoweTan:1992}.
The residual metaplectic and Virasoro phases are central scalars, hence irrelevant for
commutants. Since $\sfT$ and $\sfS$ generate $SL(2,\Z)$, the chosen lifts
$\rho(\sfT),\rho(\sfS)$ generate the lattice preimage $\widetilde\Gamma$ \emph{up to
central scalars}; as central scalars lie in the commutant of every operator, the
subgroup they generate and the full preimage have the same commutant, so
\begin{equation}
\Comm(\tS,T)=\langle\rho(\sfT),\rho(\sfS)\rangle'=\rho(\widetilde\Gamma)',
\end{equation}
where $\widetilde\Gamma$ is the inverse image of $SL(2,\Z)$ in $Mp(2,\R)$.

The even oscillator representation is not square-integrable. In the Schrodinger
model, take $\varphi_0(x)=2^{1/4}e^{-\pi x^2}$ (unit norm) and
$a(s)=\operatorname{diag}(e^s,e^{-s})$, for which
$\omega_{\mathrm{even}}(a(s))\varphi(x)=e^{s/2}\varphi(e^s x)$. Then
\begin{equation}
\langle\omega_{\mathrm{even}}(a(s))\varphi_0,\varphi_0\rangle
=\sqrt2\,e^{s/2}(1+e^{2s})^{-1/2}\sim\sqrt2\,e^{-s/2}\qquad(s\to\infty).
\end{equation}
The Gaussian spans a 1D $K$-type. Hence for
$k_1,k_2\in K$ the matrix coefficient at $k_1a(s)k_2$ differs from the displayed
coefficient only by phases, and its modulus is bi-$K$-invariant. The Haar density
on the positive chamber is $\sim e^{2s}\,ds$, so the full $KAK$ integral is
incompatible with $L^p(G)$ for $p\le4$; in particular the coefficient is not in
$L^2(G)$, hence $\omega_{\mathrm{even}}$ is not square-integrable.

\begin{theorem}[Cowling--Steger lattice restriction; Theorem~A of Bekka~\cite{Bekka:lattices}, after~\cite{CowlingSteger:1991}]\label{thm:CoS}
Let $G$ be a connected simple Lie group with finite center, $\Gamma<G$ a lattice,
and $\pi$ an irreducible unitary representation of $G$ that is not square-integrable.
Then the restriction $\pi|_\Gamma$ is irreducible.
\end{theorem}

This is precisely Theorem~A of Bekka, \emph{J.\ Funct.\ Anal.}\ \textbf{143} (1997),
33--41~\cite{Bekka:lattices}, which gives the Cowling--Steger
restriction theorem~\cite{CowlingSteger:1991} in exactly this form.
Theorem~\ref{thm:CoS} is stated for ordinary unitary representations of Lie groups and
lattices; no projective or central-character variant is needed here. The group
$Mp(2,\R)$ is connected, has simple Lie algebra, and has finite center $\Z/4$,
so it satisfies the group hypothesis. Because the center $\Z/4$ is \emph{finite},
square-integrability and square-integrability-mod-center coincide, so the
matrix-coefficient test against Haar measure on $G$ (the $\sim e^{2s}\,ds$ density
used above) is the correct one; no quotient by the center is required. The inverse
image $\widetilde\Gamma$ of $SL(2,\Z)$ is a lattice under the finite covering map.
Applying Theorem~\ref{thm:CoS} gives that
$\omega_{\mathrm{even}}|_{\widetilde\Gamma}$ is irreducible, and Schur's lemma gives
$\rho(\widetilde\Gamma)'=\C\cdot I$.

\section{Proofs of the elementary theorems and the positive-kernel proposition}\label{app:sector}

\subsection{The diagonal lemma}\label{sec:diag}

\begin{lemma}[diagonal rigidity]\label{lem:diag}
If $N=D_m$ is diagonal, $(Nf)(u)=m(u)f(u)$ with $m\in L^\infty$, and $[N,\tS]=0$,
then $m$ is a.e.\ constant.
\end{lemma}

\begin{proof}
$[D_m,\tS](u,w)=(m(u)-m(w))\,\tS(u,w)$, and
$\tS(u,w)=\sqrt2\cos(4\pi\sqrt{uw})/(uw)^{1/4}$ vanishes only on the measure-zero set
$\{\cos(4\pi\sqrt{uw})=0\}$; hence $m(u)=m(w)$ for a.e.\ $(u,w)$, so $m$ is a.e.\
constant.
\end{proof}

\begin{remark}[logical order]\label{rem:order}
The lemma is applied \emph{only} to a purely diagonal operator. For a general
$N=D_m+\sum_{n\neq0}D_{a_n}\tau_n$ the correct order is
\begin{equation}
[N,\tS]=0\ \Rightarrow\ a_n=0\ (n\neq0)\ \Rightarrow\ N=D_m\ \Rightarrow\
[D_m,\tS]=0\ \Rightarrow\ m\ \text{const}.
\end{equation}
One must \emph{not} infer $[D_m,\tS]=0$ directly from $[N,\tS]=0$: the diagonal and
off-diagonal parts of the commutator can cancel each other. The off-diagonal branches
are killed first (Sections~\ref{sec:hardedge}, \ref{sec:E}), and only then is the
diagonal lemma invoked.
\end{remark}

\subsection{Proof of Theorem~\ref{thm:A} (graph rigidity)}\label{app:graphproof}

\begin{proof}[Proof of Theorem~\ref{thm:A}]
$U_\varphi T=TU_\varphi$ means $e^{2\pi i\varphi(P)^2}f(\varphi(P))=e^{2\pi iP^2}f(\varphi(P))$
for all $f$ and a.e.\ $P$, so $\varphi(P)^2-P^2\in\Z$ a.e. In the
squared-momentum variable $u=P^2$, set $G(u)=\varphi(\sqrt u)^2$. Then
$G(u)-u\in\Z$ a.e., so $G$ preserves
$\theta:=u\bmod1$ and maps the fiber over each $\theta$ into itself. Pushing
$\rho_0(P)\,dP$ forward by $P\mapsto u=P^2$ gives $w(u)\,du$ with
$w(u)=\rho_0(\sqrt u)/(2\sqrt u)$, which disintegrates over $\theta\in[0,1)$ as
\begin{equation}\label{eq:disint}
\int f(u)\,w(u)\,du=\int_0^1\Bigl(\sum_{k\ge0,\,k+\theta>0}f(k+\theta)\,w(k+\theta)\Bigr)d\theta,
\end{equation}
i.e.\ the fiber measure over $\theta$ is the purely atomic
$\sum_k w(k+\theta)\,\delta_{k+\theta}$. As $G$ is an invertible measure-preserving
transformation preserving each fiber, the standard Rohlin disintegration theorem
\cite{Bogachev:2007} applied to the fiber decomposition gives, for a.e.\ $\theta$, a
bijection $\sigma_\theta$ of the atom set
$\{k+\theta\}_k$ with $w(\sigma_\theta(k)+\theta)=w(k+\theta)$. But $w$ is
\emph{strictly increasing}. Indeed, the numerator of
\begin{equation}
\frac{d}{dP}\left[\frac{\sinh(2\pi bP)\sinh(2\pi P/b)}{P}\right]
\end{equation}
is
\begin{equation}
\begin{aligned}
&\sinh(2\pi bP)\sinh(2\pi P/b)\\
&\quad\times
\bigl[2\pi bP\coth(2\pi bP)+2\pi(P/b)\coth(2\pi P/b)-1\bigr]>0,
\end{aligned}
\end{equation}
where we used $x\coth x>1$ for $x>0$. Thus $u\mapsto w(u)$ is injective and distinct atoms of a
fiber carry distinct masses. A mass-preserving bijection of a set with distinct masses
is the identity; hence $\sigma_\theta=\mathrm{id}$ for a.e.\ $\theta$, so $G(u)=u$ a.e.\
and $\varphi(P)=P$ a.e.
\end{proof}

\subsection{Proof of Theorem~\ref{thm:C} (hard edge)}\label{app:hardedgeproof}

The commutator kernel of a single branch is
\begin{equation}\label{eq:hesingle}
[D_a\tau_n,\tS](u,w)=a(u)\,\mathbf 1_{\{u+n>0\}}\,\tS(u+n,w)
-a(w-n)\,\mathbf 1_{\{w-n>0\}}\,\tS(u,w-n).
\end{equation}

\begin{lemma}[single-branch hard edge]\label{lem:singlebranch}
For $n\neq0$ and $a\not\equiv0$ on $\{u+n>0\}$, $[D_a\tau_n,\tS]\neq0$. No tameness is
assumed; this covers oscillatory coefficients such as $a(u)=\cos(\beta\sqrt u)$.
\end{lemma}

\begin{proof}
Use \eqref{eq:hesingle}. If $n>0$, restrict to the strip $0<w<n$: there $w-n<0$, the
second term vanishes, and the kernel equals $a(u)\,\mathbf 1_{\{u+n>0\}}\,\tS(u+n,w)$.
If $n<0$, restrict to $0<u<|n|$: there $u+n<0$, the first term vanishes, and the kernel
equals $-a(w-n)\,\tS(u,w-n)$. Either way the surviving term is a product of $a$
(nonzero on a positive-measure set, by hypothesis) and $\tS(\cdot,\cdot)$ (nonzero off
a measure-zero set), hence nonzero on a set of positive measure. So
$[D_a\tau_n,\tS]\neq0$, which is case~(i).
\end{proof}

\begin{proof}[Proof of Theorem~\ref{thm:C}, cases (ii)--(iii)]
By the reflection $n\mapsto-n$, $u\leftrightarrow w$ it suffices to treat an $F$
containing a positive element; set $n_+=\min(F\cap\Z_{>0})$. On the strip $0<w<n_+$
every positive-branch second-family term $a_n(w-n)\mathbf 1_{\{w-n>0\}}$ vanishes
($w-n<0$ for $n\ge n_+$), and the negative-branch arguments $w-n=w+|n|\ge1$ stay away
from the singularity. Multiplying the (a.e.-zero) commutator kernel by $w^{1/4}$, for
a.e.\ fixed $u$,
\begin{equation}\label{eq:AB}
A_u(w)=w^{1/4}B_u(w),\qquad
A_u(w)=\!\!\sum_{n\in F,\,u+n>0}\!\!\gamma_n(u)\cos\!\bigl(c_n(u)\sqrt w\bigr),
\end{equation}
where $\gamma_n(u)=\sqrt2\,a_n(u)/(u+n)^{1/4}$, $c_n(u)=4\pi\sqrt{u+n}$, and $B_u$
collects the negative-branch second-family terms; with $d=\#\{n\in F:u+n>0\}$, $A_u$
is entire in $w$ (each $\cos(c\sqrt w)=\sum_k(-1)^kc^{2k}w^k/(2k)!$) and $B_u\in
C^{d-1}$ near $w=0$.

In the one-sided case (ii) there are no negative branches, so $B_u\equiv0$, hence
$A_u\equiv0$ on the strip and every Taylor coefficient of $A_u$ vanishes; only
$L^\infty_{\mathrm{loc}}$ is used. In the mixed case (iii), if the lowest nonzero
Taylor coefficient of $A_u$ at $w=0$ were $Cw^r$ with $C\neq0$, $0\le r<d$, then
$B_u(w)=A_u(w)/w^{1/4}\sim Cw^{r-1/4}$, whose $r$-th derivative diverges as
$w\to0^+$, contradicting $B_u\in C^{d-1}$. Either way $A_u$ vanishes to order $\ge d$
at $w=0$, giving
\begin{equation}\label{eq:vander}
\sum_{n\in F,\,u+n>0}\gamma_n(u)\,c_n(u)^{2m}=0,\qquad m=0,\dots,d-1.
\end{equation}
The nodes $c_n(u)^2=16\pi^2(u+n)$ are distinct in $n$, so this $d\times d$ Vandermonde
system forces $\gamma_n(u)=0$, i.e.\ $a_n(u)=0$ a.e.\ on $\{u+n>0\}$.
\end{proof}

\begin{remark}[measure hygiene]
The chain is: (i) kernel uniqueness $\Rightarrow$ the commutator kernel is a.e.\ zero;
(ii) Fubini fixes a.e.\ $u$, giving the $w$-direction a.e.\ equality; (iii)
$A_u(w)-w^{1/4}B_u(w)$ is continuous on $w\in(0,n_+)$, so a.e.\ zero there
$\Rightarrow$ identically zero on the interval; (iv) only then does one take $w\to0^+$
for the Taylor / obstruction / Vandermonde step.
After Fubini we fix $u$ outside the finite threshold set $\{-n:n\in F\}$ and outside the
null set where some $a_n$ lacks its chosen Lebesgue representative; the active branch set
$\{n\in F:u+n>0\}$ is then locally constant and the Vandermonde argument applies. The
excluded set is null, so the conclusion $a_n=0$ a.e.\ is unchanged.
\end{remark}

\begin{remark}[sharp regularity, optional]
With the positive hard-edge strip used above, only the coefficients appearing on the
$B_u$ side (the negative branches, since they are evaluated as functions of
$w-n=w+|n|$) need $C^{d-1}$ regularity. The corner-side coefficients $a_n(u)$ enter
\eqref{eq:AB} as constants in $w$ for fixed $u$ and need only $L^\infty$. If the
reflected edge is used instead, the corresponding ``opposite-side'' branches are the
ones requiring this regularity. The uniform $C^{|F|-1}_{\mathrm{loc}}$ hypothesis in
(iii) is the cleaner symmetric statement.
\end{remark}

\subsection{Proofs of Proposition~\ref{prop:schur} and Corollary~\ref{cor:vacpos}}\label{app:schurproof}

\begin{proof}[Proof of Proposition~\ref{prop:schur}]
Local finiteness first: for compacts $K,K'\Subset(0,\infty)$,
\begin{equation}
N(K\times K')\le\bigl(\min\nolimits_{K'}r\bigr)^{-1}\iint_{K\times K'}r(v)\,N(du,dv)
\le\bigl(\min\nolimits_{K'}r\bigr)^{-1}\int_K r(u)\,du<\infty
\end{equation}
by the first equation of \eqref{eq:vacmeasure}, so $N$ is locally finite, hence
Radon on $(0,\infty)^2$, and Tonelli applies to the nonnegative integrands below.
Since $N$ lives on $(0,\infty)^2$, where $0<r<\infty$, we may split
$1=(r(v)/r(u))^{1/2}(r(u)/r(v))^{1/2}$ $N$-a.e.\ and apply the Cauchy--Schwarz
inequality with respect to $N$:
\begin{equation}
\begin{aligned}
\iint|g(u)||f(v)|\,N(du,dv)
&\le\Bigl(\iint|g(u)|^2\tfrac{r(v)}{r(u)}\,N(du,dv)\Bigr)^{1/2}
\Bigl(\iint|f(v)|^2\tfrac{r(u)}{r(v)}\,N(du,dv)\Bigr)^{1/2}\\
&=\|g\|_2\,\|f\|_2,
\end{aligned}
\end{equation}
the last step by \eqref{eq:vacmeasure}: the first equation, tested against
$|g(u)|^2/r(u)$, gives the $g$-factor, and the second, tested against
$|f(v)|^2/r(v)$, gives the $f$-factor. The bounded form extends uniquely from the
dense subspace $C_c$ to $L^2$.
\end{proof}

\begin{proof}[Proof of Corollary~\ref{cor:vacpos}]
Theorem~\ref{thm:layer2a} gives $N=cI$ as an operator. For $f,g\in C_c$ the form of
$cI$ is $c\int\overline{g}f\,du$; finite sums of products $\overline{g(u)}f(v)$ are
uniformly dense on compact rectangles, and a Radon measure is determined by such
pairings, so the measure equals $c\,\delta(u-v)\,du$. The first vacuum equation then
reads $c\,r(u)\,du=r(u)\,du$, and $r>0$ gives $c=1$.
\end{proof}

\section{Full proof of Theorem~\ref{thm:E} (high-energy method)}\label{app:proofs}

We work in the symmetrized realization of Section~\ref{sec:setup}: $u=P^2\in(0,\infty)$,
$\tS$ as in \eqref{eq:Stilde}, and $N=D_m+\sum_{n\neq0}D_{a_n}\tau_n$ as in
\eqref{eq:Nform}. (The finite-branch Theorem~\ref{thm:C} is proved at the hard edge in
Appendix~\ref{app:sector} and needs none of this machinery.)

\paragraph{The commutator kernel.} On the domains where the shifts are defined the
branch shift has kernels $(\tau_n\tS)(u,w)=\tS(u+n,w)$ and
$(\tS\,\tau_n)(u,w)=\tS(u,w-n)$, so
\begin{equation}\label{eq:commkernel}
[N,\tS](u,w)=(m(u)-m(w))\,\tS(u,w)
+\sum_{n\neq0}\bigl[a_n(u)\,\tS(u+n,w)-a_n(w-n)\,\tS(u,w-n)\bigr].
\end{equation}
(In the constant-coefficient case $m\equiv c_0$, $a_n\equiv c_n$, so the diagonal term
vanishes.)

\begin{lemma}[local kernel realization]\label{lem:localkernel}
Let $E=I\times J\Subset(0,\infty)^2$ be a compact rectangle, and write
$K_0(u,w)=(m(u)-m(w))\tS(u,w)$ and
\begin{equation}
K_n(u,w)=a_n(u)\mathbf 1_{\{u+n>0\}}\tS(u+n,w)
-a_n(w-n)\mathbf 1_{\{w-n>0\}}\tS(u,w-n).
\end{equation}
Under \eqref{eq:T2},
\begin{equation}
\sum_{n\neq0}\|K_n\|_{L^2(E)}\le C_E\sum_{n\neq0}\|a_n\|_\infty<\infty.
\end{equation}
Consequently $K_E:=K_0+\sum_{n\neq0}K_n$ converges in $L^2(E)$ and represents the
Schwartz kernel of the bounded commutator $[N,\tS]$ restricted to $E$: for
$f\in C_c^\infty(I)$ and $h\in C_c^\infty(J)$,
\begin{equation}
\langle f,[N,\tS]h\rangle=\iint_E \overline{f(u)}\,K_E(u,w)\,h(w)\,du\,dw .
\end{equation}
In particular, if $[N,\tS]=0$ as a bounded operator, then $K_E=0$ a.e.\ on every such
$E$.
\end{lemma}

\begin{proof}
On $E$, the kernel $\tS(x,y)$ is bounded whenever $x,y$ stay in a compact subset of
$(0,\infty)$, and near one axis its square has only the integrable singularity
$x^{-1/2}$ or $y^{-1/2}$. Thus the finitely many branch terms for which $u+n$ or
$w-n$ can approach $0$ on $E$ have finite $L^2(E)$ norm bounded by
$C_E\|a_n\|_\infty$. All remaining active terms keep both arguments away from the axes;
for large positive arguments we use $|\tS(x,y)|\le C_E x^{-1/4}$ or
$C_E y^{-1/4}$, and in particular $|\tS|\le C_E$ on $E$. Hence, uniformly in $n$,
$\|K_n\|_{L^2(E)}\le C_E\|a_n\|_\infty$, proving the $L^2(E)$ summability.

For a finite branch set $F$, the partial commutator has kernel
$K_0+\sum_{n\in F}K_n$ on $E$. The partial operators converge in operator norm because
$\sum_n\|D_{a_n}\tau_n\|\le\sum_n\|a_n\|_\infty<\infty$, while the partial kernels
converge in $L^2(E)$ by the first paragraph. Pairing against
$f\otimes h\in L^2(E)$ and passing to the limit gives the displayed kernel identity.
If $[N,\tS]=0$, all such rank-one pairings vanish. Finite sums of rank-one functions
are dense in $L^2(E)$, hence $K_E=0$ in $L^2(E)$ and therefore a.e.
\end{proof}

\paragraph{Two frequency families.} Fix $n_0\neq0$ and a base point $u$ with
$u+n_0>0$; the functional below integrates over a nearby variable $u'$, so we keep the
$u'$-dependence explicit. Multiplying \eqref{eq:commkernel} (at argument $u'$) by
$w^{1/4}$ and setting $\xi=\sqrt w$, with
$\beta(u'):=4\pi\sqrt{u'}$ and $\alpha_n(u'):=4\pi\sqrt{u'+n}$,
\begin{align}
w^{1/4}\tS(u',w)&=\tfrac{\sqrt2}{u'^{1/4}}\cos(\beta(u')\,\xi),\notag\\
w^{1/4}\tS(u'+n,w)&=\tfrac{\sqrt2}{(u'+n)^{1/4}}\cos(\alpha_n(u')\,\xi)=:A_n(\xi)
&&\text{(family A)},\\
w^{1/4}\tS(u',w-n)&=\tfrac{\sqrt2}{u'^{1/4}}\cos\!\bigl(\beta(u')\sqrt{\xi^2-n}\bigr)
\bigl(\tfrac{\xi^2}{\xi^2-n}\bigr)^{1/4}=:B_n(\xi)
&&\text{(family B, $\xi^2>n$)}\notag,
\end{align}
so that, a.e.,
\begin{equation}\label{eq:families}
w^{1/4}[N,\tS](u',w)=(m(u')-m(\xi^2))\tfrac{\sqrt2\cos(\beta(u')\xi)}{u'^{1/4}}
+\sum_{n\neq0}\bigl[a_n(u')A_n(\xi)-a_n(\xi^2-n)B_n(\xi)\bigr].
\end{equation}
Family A carries the distinct, branch-labeled frequency $\alpha_n(u')$; the diagonal
and all of family B carry instantaneous frequency $\to\beta(u')$. Since
$n_0\neq0\Rightarrow\alpha_{n_0}(u')\neq\beta(u')$, branch $n_0$ is the \emph{only} term
living at frequency $\alpha_{n_0}(u')$.

\paragraph{The tracking-frequency functional.} Fix $g\in C_c^\infty((0,\infty))$,
$g\ge0$, $\int g=1$, $\operatorname{supp}g\subset[a_0,b_0]$ with $0<a_0<b_0$, and a
mollifier $\varphi_\varepsilon(s)=\varepsilon^{-1}\varphi(s/\varepsilon)$,
$\varphi\in C_c^\infty([-1,1])$, $\int\varphi=1$; take $u\notin\Z$ and
$\varepsilon<\tfrac12\operatorname{dist}(u,\Z)$, so that
$\operatorname{supp}\varphi_\varepsilon(\cdot-u)$ avoids $\Z$, and set
$\delta_\varepsilon(u)=\inf_{\operatorname{supp}\varphi_\varepsilon}u'>0$ and
$\Delta_\varepsilon(u)=\inf\{u'+n:u'\in\operatorname{supp}\varphi_\varepsilon(\cdot-u),\,
n\in\Z\setminus\{0\},\,u'+n>0\}>0$ (both bounded below as $\varepsilon\to0$ for
$u\notin\Z$; $u\in\Z$ is null and excluded throughout). For a kernel $K(u',w)$ set
\begin{equation}\label{eq:proj}
P^{u,n_0}_{\varepsilon,T}(K):=\iint\varphi_\varepsilon(u'-u)\,\bigl[w^{1/4}K(u',w)\bigr]_{w=\xi^2}
\,\cos\!\bigl(\alpha_{n_0}(u')\,\xi\bigr)\,\frac{g(\xi/T)}{T}\,d\xi\,du'.
\end{equation}
The test frequency $\alpha_{n_0}(u')$ \emph{tracks the integration variable $u'$};
freezing it at $\alpha_{n_0}(u)$ would extract $0$ (Remark~\ref{rem:E0}). Changing back
$w=\xi^2$, $P^{u,n_0}_{\varepsilon,T}(K)=\iint K(u',w)\,\Theta_{\varepsilon,T}(u',w)\,dw\,du'$
with $\Theta_{\varepsilon,T}(u',w)=\varphi_\varepsilon(u'-u)\,w^{-1/4}\cos(\alpha_{n_0}(u')\sqrt w)\,g(\sqrt w/T)/(2T)$,
which has compact support ($u'\in\operatorname{supp}\varphi_\varepsilon$,
$w\in[a_0^2T^2,b_0^2T^2]$) and lies in $L^2$
($\|\Theta_{\varepsilon,T}\|_2^2=\|\varphi_\varepsilon\|_2^2\,\|g\|_2^2/(2T)<\infty$). On
that region $[N,\tS]$ is bounded ($\tS$ is smooth off the axes), hence in $L^2$, so
\eqref{eq:proj} is a genuine compactly supported test pairing: no Dirac slice of
the non-Hilbert--Schmidt $\tS$ is invoked. In Step~1 below we apply this pairing
first to finite branch sums; the uniform estimate of Lemma~\ref{lem:E2} then gives
absolute convergence of the resulting scalar series.

\begin{lemma}[uniform frequency projection]\label{lem:E1}
Let $K\subset(0,\infty)$ be compact and $g$ as above.
\emph{(i)} Let $\mathrm{ph}_{u'}(\xi)$ be one of: a pure cosine phase
$\alpha_n(u')\xi$ with $n\neq n_0$; the diagonal phase $\beta(u')\xi$; or a family-B
phase $\beta(u')\sqrt{\xi^2-n}$ (fixed $n$, on $\xi^2>n$). Assume the support keeps the
relevant arguments in a compact subset of $(0,\infty)$ (so $u'+n$ stays away from $0$),
the uniform separation $\inf_{u'\in K}|\alpha_{n_0}(u')-\Phi(u')|>0$ holds with
$\Phi=\alpha_n$ resp.\ $\Phi=\beta$, and $m_T$ is a symbol factor uniform in $u'$
($|m_T|\le C_0$, $|\partial_\xi m_T|\le C_1\xi^{-1}$ for large $\xi$). Then
\begin{equation}
\sup_{u'\in K}\Bigl|\tfrac1T\!\int g(\xi/T)\,m_T(\xi)\cos(\alpha_{n_0}(u')\xi)\cos(\mathrm{ph}_{u'}(\xi))\,d\xi\Bigr|\to0,
\end{equation}
covering family-A branches $n\neq n_0$ ($\Phi=\alpha_n$) and family B / the diagonal
($\Phi=\beta$).
\emph{(ii)} $\displaystyle\sup_{u'\in K}\bigl|\tfrac1T\!\int g(\xi/T)\cos^2(\alpha_{n_0}(u')\xi)\,d\xi-\tfrac12\bigr|\to0$.
\end{lemma}

\begin{proof}
Write $\cos\cdot\cos=\tfrac12[\cos(\text{sum})+\cos(\text{diff})]$; each combined phase
$\psi$ has $|\psi'|\ge\delta:=\inf_K\min(|\alpha_{n_0}-\Phi|,\,\alpha_{n_0}+\Phi)>0$ on
$\operatorname{supp}g(\cdot/T)=[a_0T,b_0T]$ for $T$ large, and $|\psi''|=O(\xi^{-3})$
(family B) or $0$ (pure cosines). One integration by parts gives
$|\int A_Te^{i\psi}|\le\delta^{-1}\!\int|A_T'|+\delta^{-2}\!\int|A_T||\psi''|=O(1)$
uniformly in $u'\in K$, so the $\tfrac1T$-average $\to0$ uniformly. For (ii),
$\cos^2=\tfrac12(1+\cos2\alpha_{n_0}\xi)$: the oscillatory half $\to0$ uniformly, the
constant half gives $\tfrac12$.
\end{proof}

\begin{lemma}[uniform in $(n,T)$, stable as $\varepsilon\to0$]\label{lem:E2}
There is $C=C(u,n_0,g)$, independent of $n$, of $T$, and of $\varepsilon$ down to
$\varepsilon\to0$, with
\begin{equation}
\bigl|P^{u,n_0}_{\varepsilon,T}([D_{a_n}\tau_n,\tS])\bigr|\le C\,\|a_n\|_\infty
\qquad(n\neq0,\ T\ \text{large}).
\end{equation}
\end{lemma}

\begin{proof}
The $w^{1/4}$-kernel of the branch is
\begin{equation}
a_n(u')A_n(\xi)-a_n(\xi^2-n)B_n(\xi).
\end{equation}
\emph{Family A (uniform in $n$).} $|a_n(u')A_n(\xi)|\le\sqrt2\,\|a_n\|_\infty(u'+n)^{-1/4}$
and $\tfrac1T\int g(\xi/T)|\cos|\,d\xi\le\|g\|_1$. The prefactor $(u'+n)^{-1/4}$ is largest
on the single near-edge active branch and is bounded there by $\Delta_\varepsilon(u)^{-1/4}$,
uniformly in $n$; integrating $\varphi_\varepsilon$ over $u'$ gives
$\le\sqrt2\,\|g\|_1\,\Delta_\varepsilon(u)^{-1/4}\|a_n\|_\infty$.
\emph{Family B (uniform in $n$).} Substitute $\xi=Ty$, $y\in[a_0,b_0]$,
$\lambda:=n/T^2$:
\begin{equation}
\begin{aligned}
\tfrac1T\Bigl|\int g\,a_n(\xi^2-n)B_n\cos\Bigr|
&\le\tfrac{\sqrt2}{u'^{1/4}}\|a_n\|_\infty
\int_{a_0}^{b_0}g(y)\bigl(\tfrac{y^2}{y^2-\lambda}\bigr)_+^{1/4}dy\\
&\le\tfrac{\sqrt2}{u'^{1/4}}M(g)\|a_n\|_\infty,
\end{aligned}
\end{equation}
$M(g):=\sup_{\lambda\in\R}\int_{a_0}^{b_0}g(y)(y^2/(y^2-\lambda))_+^{1/4}dy<\infty$ (for
$\lambda<0$ the ratio $<1$; for $\lambda>b_0^2$ the support is empty; for
$\lambda\in[a_0^2,b_0^2]$ the only singularity $y=\sqrt\lambda$ gives integrand
$\sim|y-\sqrt\lambda|^{-1/4}$, integrable and uniform on the compact $\lambda$-range).
Integrating $\varphi_\varepsilon$ gives $\le\sqrt2\,\delta_\varepsilon(u)^{-1/4}M(g)\|a_n\|_\infty$.
So $C_\varepsilon=\sqrt2\bigl(\|g\|_1\,\Delta_\varepsilon(u)^{-1/4}+M(g)\,\delta_\varepsilon(u)^{-1/4}\bigr)$,
independent of $n$ and $T$; for $u\notin\Z$ both $\Delta_\varepsilon(u)$ and
$\delta_\varepsilon(u)$ stay bounded below as $\varepsilon\to0$, so $C_\varepsilon$ stays
bounded.
\end{proof}

\begin{proof}[Proof of Theorem~\ref{thm:E}]
\emph{Step 0 (kernel realization).} By Lemma~\ref{lem:localkernel}, on every compact
rectangle away from the axes the formal kernel \eqref{eq:commkernel} converges in
local $L^2$ and represents the bounded commutator $[N,\tS]$. Since $[N,\tS]=0$, this
local kernel is a.e.\ zero; hence
$P^{u,n_0}_{\varepsilon,T}([N,\tS])=0$ for every $\varepsilon,T$ and every base point
$u$ for which the support of $\varphi_\varepsilon(\cdot-u)$ stays inside
$\{u'>0,\ u'+n_0>0\}$.

\emph{Step 1 (fix $\varepsilon$; expand and exchange).} Apply
$P^{u,n_0}_{\varepsilon,T}$ first to finite branch sums. Lemma~\ref{lem:localkernel}
identifies the $L^2$ limit of these local kernels with the kernel of $[N,\tS]$ on
$\operatorname{supp}\Theta_{\varepsilon,T}$. For each fixed $\varepsilon,T$,
Lemma~\ref{lem:E2} gives
\begin{equation}
\sum_{n\neq0}\bigl|P^{u,n_0}_{\varepsilon,T}([D_{a_n}\tau_n,\tS])\bigr|
\le C\sum_{n\neq0}\|a_n\|_\infty<\infty,
\end{equation}
with $C$ independent of $T$ and stable as $\varepsilon\to0$. Hence the projected
finite branch sums converge to the projected infinite branch part as an absolutely
convergent scalar series; no pointwise domination of the full non-Hilbert--Schmidt
kernel is needed. Since $P^{u,n_0}_{\varepsilon,T}([N,\tS])=0$, letting
$T\to\infty$ and using the same summable bound for dominated convergence (counting
measure on $n$) gives
\begin{equation}
0=\lim_T P_{\varepsilon,T}([D_m,\tS])+\sum_{n\neq0}\lim_T P_{\varepsilon,T}([D_{a_n}\tau_n,\tS]).
\end{equation}
By Lemma~\ref{lem:E1}: the diagonal (frequency $\beta(u')\neq\alpha_{n_0}(u')$; $m$
tame, so $m(\xi^2)$ is a symbol factor) $\to0$; every branch $n\neq n_0$ $\to0$; and
the family-B half of branch $n_0$ $\to0$. Only the family-A half of branch $n_0$
survives, via \ref{lem:E1}(ii):
\begin{equation}
0=\tfrac{\sqrt2}{2}\int\varphi_\varepsilon(u'-u)\,\frac{a_{n_0}(u')}{(u'+n_0)^{1/4}}\,du'.
\end{equation}

\emph{Step 2 ($\varepsilon\to0$).} $u'\mapsto a_{n_0}(u')(u'+n_0)^{-1/4}\in
L^1_{\mathrm{loc}}$ near $u$; at every Lebesgue point $u$ with $u\notin\Z$,
$0=\tfrac{\sqrt2}{2}\,a_{n_0}(u)(u+n_0)^{-1/4}$, so $a_{n_0}(u)=0$. A.e.\ $u$ is a
Lebesgue point and the excluded integer set $\Z$ is null, so $a_{n_0}=0$ a.e.\ on
$\{u+n_0>0\}$. As $n_0\neq0$ was arbitrary, all off-diagonal branches vanish.

\emph{Step 3 (diagonal, last).} Only now is $N=D_m$, so
$[D_m,\tS]=(m(u)-m(w))\tS=0$ a.e.; $\tS\neq0$ a.e.\ forces $m$ constant
(Lemma~\ref{lem:diag}). Hence $N=c\cdot I$. The operator-ideal corollary follows
because no nonzero scalar identity lies in the stated ideal.
\end{proof}

\begin{remark}[why tracking, not freezing]\label{rem:E0}
With a \emph{frozen} test frequency $\alpha_{n_0}(u)$, the inner $\tfrac1T$-average
against family A of branch $n_0$ is
$\tfrac1T\int g\,\cos(\alpha_{n_0}(u')\xi)\cos(\alpha_{n_0}(u)\xi)\,d\xi
\to\tfrac12\,\mathbf 1_{\{\alpha_{n_0}(u')=\alpha_{n_0}(u)\}}=\tfrac12\,\mathbf 1_{\{u'=u\}}$,
supported on a Lebesgue-null set in $u'$; integrating $\varphi_\varepsilon$ then yields
$0$, so the functional would extract $0$ rather than $a_{n_0}(u)$. Tracking
$\alpha_{n_0}(u')$ matches branch $n_0$ at every $u'$ (Lemma~\ref{lem:E1}(ii)), giving
$\tfrac12\,a_{n_0}(u')(u'+n_0)^{-1/4}$ pointwise, whose mollified average
$\to\tfrac12\,a_{n_0}(u)(u+n_0)^{-1/4}$. The tracking is essential, not cosmetic.
\end{remark}

\paragraph{Seminorm requirements.} Lemma~\ref{lem:E1}'s single integration by parts uses
(T1) with $J=1$ on each $a_n$, on $m$, and on the symbol factor
$(\xi^2/(\xi^2-n))^{1/4}$; Lemma~\ref{lem:E2} and the dominated-convergence exchange
use only the sup-norm tail (T2). No higher or $n$-weighted seminorm enters.

\paragraph{Logical order.} One must \emph{not} infer $[D_m,\tS]=0$ directly from
$[N,\tS]=0$: the projection only shows the diagonal contributes nothing \emph{at
frequency $\alpha_{n_0}$}; the diagonal is forced constant only in Step 3, after every
off-diagonal branch has been killed. Correct order:
$[N,\tS]=0\Rightarrow a_n=0\ (n\neq0)\Rightarrow N=D_m\Rightarrow[D_m,\tS]=0\Rightarrow m\ \text{const}$.

\section{The quantitative layer (\texorpdfstring{$\kappa_T$}{kappa T})}\label{app:gap}

\emph{Apart from two self-contained side statements (Lemma~\ref{lem:fullline} and
Proposition~\ref{prop:hsgap}), this appendix is evidence and open-problem
discussion; nothing in it is used in the main theorems.} Exact scalarity of the
full bounded commutant is Theorem~\ref{thm:layer2a};
the elementary sector theorems (Theorems~\ref{thm:C}, \ref{thm:E}) give independent
mechanisms. Both exact statements should be distinguished from the quantitative gap,
\begin{equation}
\kappa_T=\inf\bigl\{\,\|[S,N]\|:\ N\in\Comm(T),\ \mathrm{dist}(N,\C1)=1\,\bigr\}\ \overset{?}{>}\ 0.
\end{equation}
High-energy probes show: single integer shifts $u\to u+n$ are not approximate
commutants; finite multi-branch combinations can cancel in \emph{selected} scale
channels under an $L^2$/coefficient norm, but this does not survive a dense-scale /
wide-band validation (it is overfitting). Under the operator (multiplier) norm
$\|A\|_\infty$ the obstruction is restored and explained by the following.

\begin{lemma}[full-line reciprocal obstruction]\label{lem:fullline}
For continuous $1$-periodic zero-mean $A$,
\begin{equation}
\sup_{x>0}|A(x)-A(1/x)|\ge\tfrac12\,\|A\|_\infty.
\end{equation}
\end{lemma}

\begin{proof}
Let $D=\sup_{x>0}|A(x)-A(1/x)|$. Fix $y\in[0,1)$ and take $x=n+y$ with
$n\in\Z_{>0}$: periodicity gives $A(x)=A(y)$, while $1/x\to0$ and continuity give
$A(1/x)\to A(0)$ as $n\to\infty$. Hence $|A(y)-A(0)|\le D$ for every $y$, i.e.\
$\|A-A(0)\|_\infty\le D$. The zero mean gives
$A(0)=-\int_0^1\bigl(A(y)-A(0)\bigr)dy$, so $|A(0)|\le\|A-A(0)\|_\infty\le D$.
Therefore $\|A\|_\infty\le\|A-A(0)\|_\infty+|A(0)|\le2D$.
\end{proof}

Since $\|A\|_\infty$ is the relevant multiplier norm (the operator norm in the
bilateral/high-energy shift model), Lemma~\ref{lem:fullline} is a
degree-independent symbol-level obstruction to $\kappa_T=0$. The evidence thus
points to an operator-norm gap $\kappa_T>0$ in the tested symbol/tame classes, but
lifting the symbol obstruction to the true operator (uniform symbol-error control)
is open. A finite-dimensional Galerkin computation of regular truncated joint
commutants returns a 1D nullspace (the identity) in every truncation,
basis and quadrature; after Theorem~\ref{thm:layer2a}, this checks the
proven exact rigidity in regular Galerkin truncations, not an ingredient of the proof.

\paragraph{A Hilbert--Schmidt gap for the joint problem.} The gap $\kappa_T$
concerns $T$-exact deformations in the operator norm. In the Hilbert--Schmidt
metric, where the Weyl operator--symbol correspondence is an exact isometry, the
\emph{joint} approximate-commutant question is not open; it has a clean positive
answer by standard representation theory.

\begin{proposition}[Hilbert--Schmidt gap for the conjugation representation]\label{prop:hsgap}
Let $\mathcal C_{2}$ be the Hilbert--Schmidt class on
$\mathcal H\cong L^2_{\mathrm{even}}(\R)$. There is a constant
$c_{\mathrm{HS}}>0$ such that every $X\in\mathcal C_{2}$ satisfies
\begin{equation}
\max\bigl\{\|[\tS,X]\|_{2},\,\|[T,X]\|_{2}\bigr\}\ \ge\ c_{\mathrm{HS}}\,\|X\|_{2}.
\end{equation}
In particular there are no Hilbert--Schmidt approximate-invariant (Weyl)
sequences for the joint commutant: no sequence of unit Hilbert--Schmidt
operators asymptotically commutes with both $\tS$ and $T$.
\end{proposition}

\begin{proof}
Since $\tS,T$ are unitary, $\|\tS X\tS^{*}-X\|_{2}=\|[\tS,X]\|_{2}$ and likewise
for $T$, so the claim is a spectral-gap statement for the conjugation
representation $\Pi(\gamma)X=\rho(\gamma)X\rho(\gamma)^{*}$, which factors
through $\Gamma=PSL(2,\Z)$ (conjugation kills the center and all phases). Under
$\mathcal C_{2}\simeq\mathcal H\otimes\overline{\mathcal H}$ one has
$\Pi\simeq\omega_{\mathrm{even}}\otimes\overline{\omega_{\mathrm{even}}}$;
equivalently, in the Weyl calculus $\Pi$ is the pullback action of linear
symplectic maps on Weyl symbols (exact metaplectic covariance), the even--even
corner being an invariant direct summand. The Gaussian computation of
Appendix~\ref{app:weilproof} displays the leading exponent $e^{-s/2}$; the same
dilation estimate in the Hermite model,
$\langle\omega(a(s))f,h\rangle=e^{-s/2}\int f(y)\,\overline{h(e^{-s}y)}\,dy$
with $f,h$ Schwartz, gives $O(P(s)\,e^{-s/2})$ for \emph{all} $K$-finite
coefficients, hence $O(P(s)\,e^{-s})$ for the coefficients of the tensor square.
Against the Cartan Haar density $\asymp e^{2s}\,ds$ these lie in
$L^{2+\epsilon}(G)$ for every $\epsilon>0$, so by the Cowling--Haagerup--Howe
almost-$L^{2}$ criterion~\cite{CHH:1988} $\Pi$ is tempered:
$\Pi\prec\lambda_{G}$. Weak containment restricts to the lattice (finite subsets
of $\Gamma$ are compact in $G$), and for a measurable fundamental domain $F$
(finite covolume suffices; compactness is not used)
$L^{2}(G)\simeq\ell^{2}(\Gamma)\otimes L^{2}(F)$ gives
$\lambda_{G}|_{\Gamma}\simeq\lambda_{\Gamma}\otimes1$, hence
$\Pi|_{\Gamma}\prec\lambda_{\Gamma}$. Since $PSL(2,\Z)\simeq\Z/2*\Z/3$ is
nonamenable, $\lambda_{\Gamma}$ admits no almost-invariant vectors
\cite{BHV:2008}, and neither does any representation weakly contained in it; for
a finitely generated group this yields a uniform gap over any finite generating
set. The images of $\sfS$ and $\sfT$ generate $PSL(2,\Z)$; $\sfS$ is an involution and
$\|\Pi(\sfT^{-1})v-v\|=\|\Pi(\sfT)v-v\|$, so the symmetric generating set is
controlled, giving the displayed estimate.
\end{proof}

\begin{remark}[why this is not a Hilbert--Schmidt $\kappa_T$]\label{rem:noHSkappa}
The proposition is a gap for the \emph{joint} approximate commutant, not a
$T$-exact statement: the exact $\Comm(T)$ contains no nonzero Hilbert--Schmidt
operator at all. Indeed, if $X\in\mathcal C_{2}$ has kernel $K(u,v)$ and
$[X,T]=0$, then $(e^{2\pi iu}-e^{2\pi iv})K(u,v)=0$ a.e., so $K$ is supported on
the null set $\bigcup_{n}\{u-v=n\}$ and vanishes in $L^{2}$. $T$-exact
deformations are thus singular with respect to the Hilbert--Schmidt metric,
which localizes the difficulty of $\kappa_T$ in the non-Hilbert--Schmidt
(multiplier-type) directions, where Lemma~\ref{lem:fullline} is the symbol-level
obstruction. The finite Galerkin/Frobenius gaps of Section~\ref{sec:numerics}
are fiber-norm diagnostics, not continuum Hilbert--Schmidt statements; the
operator-norm $\kappa_T$ remains open.
\end{remark}

\printbibliography[heading=bibliography]

\end{document}